\documentclass[pdflatex,sn-mathphys-num]{sn-jnl}

\usepackage{graphicx}%
\usepackage{multirow}%
\usepackage{amsmath,amssymb,amsfonts}%
\usepackage{amsthm}%
\usepackage{mathrsfs}%
\usepackage[title]{appendix}%
\usepackage{xcolor}%
\usepackage{textcomp}%
\usepackage{manyfoot}%
\usepackage{booktabs}%
\usepackage[ruled]{algorithm2e}
\usepackage{algpseudocode}%
\usepackage{listings}%

\usepackage{hyperref}
\usepackage{amsmath}
\usepackage{mathrsfs}
\usepackage{txfonts}
\usepackage{cases}
\usepackage{float}
\usepackage{threeparttable}
\usepackage{amssymb}

\theoremstyle{thmstyleone}%
\newtheorem{theorem}{Theorem}
\theoremstyle{thmstyletwo}%
\newtheorem{example}{Example}%
\newtheorem{remark}{Remark}%
\newtheorem{lem}{Lemma}
\theoremstyle{thmstylethree}%
\newtheorem{definition}{Definition}%

\usepackage{amsmath}
\numberwithin{figure}{section}     
\usepackage[labelfont=bf,labelsep=period]{caption}
\usepackage{booktabs} 
\usepackage{float}    
\usepackage[utf8]{inputenc} 
\usepackage[T1]{fontenc}

\begin{document}

\title[Dynamic Automated Deduction by Contradiction Separation: The Standard Extension Algorithm]{Dynamic Automated Deduction by Contradiction Separation: The Standard Extension Algorithm}


\author*[1,3]{Yang Xu\footnote{\tiny Corresponding authors.\\
All the five authors are co-first authors.\\
E-mail addresses: xuyang@swjtu.edu.cn (Y. Xu), x.he@swjtu.edu.cn (X. He), swchen@swjtu.edu.cn (S. Chen), j.liu@ulster.ac.uk (J. Liu), zhongxm2013@swjtu.edu.cn (X. Zhong)}}
\author*[1,3]{Xingxing He}
\author[1,3]{Shuwei Chen}
\author[2,3]{Jun Liu}
\author[1,3]{Xiaomei Zhong}
\affil*[1]{School of Mathematics, Southwest Jiaotong University, Chengdu, 610031, China}
\affil[2]{School of Computing, Ulster University, Belfast BT15 1ED, Northern Ireland, UK}
\affil[3]{National-Local Joint Engineering Laboratory of System Credibility Automatic Verification, Southwest Jiaotong University, Chengdu, Sichuan, China}

\abstract{
Automated deduction seeks to enable machines to reason with mathematical precision and logical completeness. Classical resolution-based systems, such as Prover9, E, and Vampire, rely on binary inference, which inherently limits multi-clause synergy during proof search. The Contradiction Separation Extension (CSE) framework, introduced by Xu et al. (2018), overcame this theoretical limitation by extending deduction beyond binary inference. However, the original work did not specify how contradictions are algorithmically constructed and extended in practice. This paper presents the Standard Extension algorithm, the first explicit procedural realization of contradiction separation reasoning. The proposed method dynamically constructs contradictions through complementary literal extension, thereby operationalizing the CSE theory within a unified algorithm for satisfiability and unsatisfiability checking. The algorithm’s soundness and completeness are formally proven, and its effectiveness is supported indirectly through the performance of CSE-based systems, including CSE, CSE\_E, CSI\_E, and CSI\_Enig in major automated reasoning competitions (CASC) in the last few years. These results confirm that the Standard Extension mechanism constitutes a robust and practically validated foundation for dynamic, multi-clause automated deduction.

 }

\keywords{Automated deduction; Contradiction separation; Standard Extension algorithm; Multi-clause reasoning; Theorem proving; Dynamic deduction; CSE framework}



\maketitle

\section{\label{sec:level1}Introduction}
Automated reasoning, the use of formal logic to enable computers to perform human-like deduction, represents a cornerstone of computer science and artificial intelligence. Within this domain, automated theorem proving (ATP) seeks to transform logical reasoning into algorithmic procedures capable of autonomously establishing theorems within computational frameworks. Since McCune’s landmark proof of the Robbins conjecture using the EQP prover, ATP has evolved into a powerful paradigm for addressing open problems in mathematics, formal logic, and the verification of hardware, software, and circuit designs.

Following Robinson’s introduction of the resolution principle, binary-resolution-based provers such as Prover9, Vampire, E, and SPASS have become dominant in first-order logic reasoning. Over decades, numerous refinements, such as set-of-support strategies, ordered resolution, and subsumption optimization, have improved their performance. Nevertheless, these approaches remain fundamentally binary, combining only two clauses per inference. This inherent restriction limits the exploitation of synergistic relationships among multiple clauses, often resulting in redundant reasoning paths and unnecessarily lengthy derivations.

To address these inherent limitations, Xu \textit{et al}. (2018) \cite{xu2018contradiction} proposed the Contradiction Separation (CS) based Dynamic Multi-Clause Synergized Automated Deduction theoretical framework, called CSE framework, where the CS principle is the core logical inference rule to define how contradictions can be separated among clauses: instead of resolving only two complementary literals (as in binary resolution), CS identifies and separates multi-clause sets that collectively form a contradiction. It treats contradictions as cooperative clause structures, not just binary pairs. CSE mean the extended and systematized version of CS for automated deduction, here “E” meant “Extension”, which was added to emphasize that the 2018 work extends the basic CS idea from a theoretical notion into a full dynamic, multi-clause, synergized deduction theory (with soundness, completeness, and implementation architecture). 

However, the 2018 paper focused on the theory of contradiction separation and did not detail the specific algorithmic mechanism by which contradictions are constructed, extended, and separated during proof search. In practice, the success of subsequent CSE-based systems has relied heavily on a yet-unpublished internal algorithm that operationalizes this theory. The present paper fills this gap by introducing the Dynamic Deduction by Contradiction Separation Based on Standard Extension method, a concrete, stepwise algorithm for contradiction construction that realizes the theoretical principles of CSE.

The proposed Standard Extension (SE) method provides a systematic approach to dynamically generating and extending contradictions through complementary literal expansion. By defining standard contradictions and standard extensions, the paper formalizes how contradictions can be built incrementally and resolved efficiently within a dynamic multi-clause environment. This algorithm thus represents the first explicit procedural realization of contradiction separation, serving as the algorithmic foundation and implemented in provers such as CSE, CSE\_E, CSI\_E, and CSI\_Enig in major automated reasoning competitions (CASC) in the last few years.

The main contributions are summarized below:

\begin{itemize}
 \item Algorithmic Realization of Contradiction Separation:
Presents the first detailed algorithm (Standard Extension) that operationalizes the contradiction separation framework proposed in 2018, which bridges theoretical contradiction separation with practical, system-level reasoning efficiency.
 \item Formalization of Standard Contradictions and Extensions:
Defines a systematic method for constructing contradictions through literal extension, ensuring both soundness and completeness.
 \item Unified Dynamic Deduction Procedure:
Introduces an integrated algorithm for satisfiability and unsatisfiability checking using standard contradiction construction. Notably, if a formula is satisfiable, the method can explicitly construct a satisfiable instance—an ability absent in most traditional ATP systems
 \item Theoretical and Empirical Continuity:
Establishes the Standard Extension algorithm as the practical basis of later provers (CSE, CSE-E, CSI-E, CSI-Enig), whose success indirectly validates the proposed approach.
\end{itemize}

The remainder of the paper is organized as follows: Section~\ref{sec:level2} reviews the related works.  Section~\ref{sec:level3} overviews preliminary definitions and results. Section~\ref{sec:level4} details contradiction separation based on standard extension algorithm in propositional logic. Section~\ref{sec:level5} extends the method to first-order logic. Section~\ref{sec:level6} presents experimental evidences and system relevance.  Section~\ref{sec:level7} presents discussions, concluding remarks and future direction.

\section{\label{sec:level2}Related Works}

Automated deduction is grounded in formal logic systems and inference rules that ensure sound and complete reasoning. Classical frameworks include resolution~\cite{robinson1965machine}, sequent calculus~\cite{hilbert1999principles}, natural deduction~\cite{pelletier1998natural}, the connection method~\cite{bibel1981matrices}, term rewriting~\cite{bachmair1994rewrite,robinson2001handbook}, and mathematical induction~\cite{boyer1995boyer}. Among these, \textit{resolution} remains the most widely adopted inference mechanism in modern theorem provers due to its simplicity and completeness. However, its naïve binary application, where only two clauses participate in each inference—often produces redundant or non-contributory clauses, leading to exponential growth in search space.

To address efficiency bottlenecks, decades of research have refined resolution-based theorem proving along several complementary directions. (1) Clause simplification and redundancy elimination, including tautology removal and subsumption filtering~\cite{liang1973symbolic}, as well as model elimination techniques~\cite{loveland1969simplified,stickel1986schubert}; (2) Constrained inference strategies, such as linear, locked, semantic, hyper-, chain, and set-of-support resolution~\cite{loveland1970linear,boyer1971locking,slagle1967automatic,wos1965efficiency,robinson1965automatic}; (3) Efficient operational restrictions, including unit, input, and ordered resolution, which underpin provers such as OTTER, Isabelle, and Vampire~\cite{chang1970unit}; (4) Hybrid calculi integrating paramodulation, superposition, rewriting, or natural deduction to manage equality and guide inference~\cite{robinson1983paramodulation,bachmair1998equational,bachmair1994rewrite,pelletier1998natural}; and (5) Heuristic and strategy-driven search optimization, including dynamic clause and literal selection, axiom filtering, and adaptive strategy scheduling~\cite{portoraro1998strategic,schulz2013system,urban2010evaluation,urban2007malarea}. These advances have shaped state-of-the-art ATP systems, Vampire, SPASS, E, and Prover9, and established the basis for benchmarking through the CADE ATP System Competitions (CASC) and the TPTP problem library~\cite{sutcliffe2016cade,sutcliffe2017tptp}.

Despite these achievements, binary resolution remains structurally limited. Its pairwise clause interactions restrict synergy among multiple premises, obscure goal directionality, and often generate large sets of irrelevant clauses. Most enhancements address local efficiency but do not fundamentally overcome the binary constraint. Consequently, current provers can struggle with scalability in large-theory reasoning and neurosymbolic ATP applications, where effective clause cooperation and dynamic inference control are essential~\cite{Ospanov2025APOLLO,Hu2025HybridProver,Ambati2025ProofNet,Chen2025SeedProver}.

Several extensions have attempted to broaden resolution’s expressive power. Hyper-resolution, semantic resolution, and model elimination introduced new inference patterns, but these remain static and fail to capture dynamic multi-clause interactions. More recent saturation-based reasoning, notably in E~\cite{schulz2013system} and Vampire~\cite{kovacs2013first}, achieves high efficiency using the superposition calculus for equality reasoning. Concurrently, the rise of machine-learning-guided reasoning has opened a complementary research line: systems such as ENIGMA~\cite{Jakubuv2020ENIGMA} integrate statistical learning into clause selection, while neural-guided provers like DeepHOL and Graph2Seq embed logical structures into learned vector spaces. Yet, these approaches still depend on resolution or superposition as their underlying inference calculus and thus inherit the intrinsic binary limitations.

The \textit{Contradiction Separation (CS)} principle proposed by Xu et al.~\cite{xu2018contradiction} introduced a fundamentally different perspective. Instead of resolving complementary literals pairwise, CS treats contradictions as cooperative clause structures that can be dynamically constructed and separated during deduction. The Contradiction Separation–based Dynamic Multi-Clause Synergized Automated Deduction (CSE) framework provided the first theoretical formulation of such multi-clause reasoning but did not explicitly describe how contradictions are algorithmically generated. Subsequent studies~\cite{cao2019first,cao2019contradiction,cao2021multi,chen2018look,chen2020clause,he2018new,he2022structures,zhong2020novel} explored extensions of contradiction separation, yet challenges remain: the construction of contradictions is often ad hoc, clause-selection efficiency varies across problem domains, and a unified constructive mechanism for both satisfiability and unsatisfiability checking is still lacking.

The present work addresses these limitations through the \textbf{Standard Extension (SE) algorithm}, which concretizes contradiction separation into a stepwise, dynamically extensible procedure. The SE method forms the algorithmic basis of modern contradiction-separation-based systems such as CSE, CSE\_E, CSI\_E, and CSI\_Enig, which have demonstrated strong empirical performance in the CASC competitions~\cite{cao2019first,cao2021multi}. 

In parallel, the broader ATP community has increasingly explored \textit{neurosymbolic reasoning}, the integration of symbolic logic with machine learning techniques for proof guidance and clause selection. Recent systems such as APOLLO, HybridProver, ProofNet, and SeedProver~\cite{Ospanov2025APOLLO,Hu2025HybridProver,Ambati2025ProofNet,Chen2025SeedProver} exemplify this trend, employing neural networks to guide inference selection while preserving symbolic soundness. Although these systems are not directly based on Contradiction Separation (CSE), their design motivations align closely with its principles: both seek to combine structural expressiveness with data-driven adaptability. In this context, the CSE framework and its algorithmic extensions, such as the Standard Extension (SE), provide a promising foundation for future neurosymbolic integration.

\section{\label{sec:level3}Preliminaries}
In this section, we briefly recall some preliminary definitions and results for dynamic deduction of contradiction  separation based on propositional logic and first-order logic. We refer the reader to, for instance, \cite{robinson1965automatic,xu2016novel,xu2018contradiction} , for more details about logical notations and resolution concepts.

\begin{definition} \cite{xu2018contradiction}
Let $S = \{C_1, C_2, \ldots, C_m\}$ be a clauses set. The Cartesian
product of $C_1, C_2, \ldots, C_m$, denoted by $\prod^m_{i=1}C_i$,
is the set of all ordered tuples $(p_1, p_2, \ldots, p_m)$ such
that $p_i \in C_i \,(i = 1, 2, \ldots, m)$, where $p_i$ is a literal
and $C_i$ is also regarded as a set of literals.
\end{definition}

\begin{definition} \cite{xu2018contradiction} \label{def3}
Let $S = \{C_1, C_2,  \ldots, C_m\}$ be a clauses set. If for all
$(p_1, p_2, \ldots, p_m)$ $\in \prod^{m}_{i=1}C_i$, there exists at
least one complementary pair among $(p_1, p_2, \ldots, p_m)$, then
$S = \bigwedge^m_{i=1}C_i$ is called a standard contradiction. If
$\bigwedge^m_{i=1}C_i$ is unsatisfiable, then $S =
\bigwedge^m_{i=1}C_i$ is called a quasi-contradiction.
\end{definition}

\begin{lem}\cite{xu2018contradiction}\label{lem3.1}
Assume a clauses set $S = \{C_1, C_2, \cdots, C_m\}$ in
propositional logic. Then $S$ is a standard contradiction if and
only if $S$ is a quasi-contradiction.
\end{lem}

\begin{definition} \cite{xu2018contradiction} \label{def3.2}
Assume a clauses set $S = \{C_1, C_2, \ldots, C_m\}$ in
propositional logic. The following inference rule that produces a
new clause from $S$ is called a contradiction separation rule, in
short, a CS rule:

For each $C_i \,(i = 1, 2, \ldots, m)$, we separate it into two
sub-clauses $C_i^{-}$ and $C_i^{+}$ such that
\begin{enumerate}
\item[(1)] $C_i = C_i^{-} \vee C_i^{+}$, where $C_i^{-}$ and $C_i^{+}$ have no common
literals.
\item[(2)] $C_i^{+}$ may be an empty clause itself, but $C_i^{-}$ cannot
be an empty clause.
\item[(3)] $\bigwedge_{i=1}^m C_i^{-}$ is a standard contradiction (SC).
\end{enumerate}
The resulting clause $\bigvee_{i=1}^m C_i^{+}$, denoted as
$\mathscr{C}_m(C_1, C_2, \ldots, C_m)$, is called a contradiction
separation clause (CSC) of $C_ 1, C_ 2, \ldots, C _m $.
\end{definition}

\begin{definition}\label{Def2.2.2}\cite{xu2018contradiction} 
Suppose a clause set $S =\{C_1 , C_2 , \cdots, C_m \}$ in propositional logic. $\phi_1 , \phi_2 ,\cdots, \phi_t$ is called a contradiction separation based dynamic deduction sequence (or a CS based dynamic deduction sequence) from $S$ to a clause $\phi_t$. If $\phi_ i (i = 1, 2, \cdots, t)$ satisfies:
\begin{enumerate}
\item[(1)] $\phi_i \in S$, or
\item[(2)] there exist $r_1, r_2,\cdots, r_{k_i} < i$, $\phi_i =\mathscr{C}_{k_i}(\phi_{r_1}, \phi_{r_2}, \cdots, \phi_{r_{k_i}})$.
\end{enumerate}
\end{definition}

\begin{theorem}\cite{xu2018contradiction}
Suppose a clauses set $S = \{C_1, C_2, \ldots, C_m\}$ in
propositional logic. $\Phi_1, \Phi_2, \ldots, \Phi_t$ is a CS based
deduction from $S$ to a clause $\Phi_t$. If $\Phi_t$ is an empty
clause, then $S$ is unsatisfiable.
\end{theorem}

\begin{theorem}\cite{xu2018contradiction}
Suppose a clauses set $S = \{C_1, C_2, \ldots, C_m\}$ in
propositional logic. If $S$ is unsatisfiable, then there exists a CS
based deduction from $S$ to an empty clause.
\end{theorem}

\begin{definition}\label{Def3.4}\cite{xu2018contradiction}
\textbf{(Standard Contradiction Separation Rule in First-Order Logic) }Suppose a clauses set $S = \{C_1, C_2,\cdots, C_m\}$ in first-order logic. Without loss of generality, assume that there does not exist the same variables among $C_1, C_2,\cdots, C_m$ (if there exists the same variables, there exists a rename substitution which makes them different). The following inference rule that produces a new clause from $S$ is called a standard contradiction separation rule, in short, an S-CS rule:

For each $C_i (i =1, 2,\cdots, m)$, firstly applying a substitution $\sigma_i$ to $C_i$ ($\sigma_i$ could be an empty substitution but not necessary the most general unifier), denoted as $C_i^{\sigma_i}$. Then separate $C_i^{\sigma_i}$ into two sub-clauses $C_i^{\sigma_i^-}$ and $C_i^{\sigma_i^+}$ such that
\begin{enumerate}
\item[(i)]  $C_i^{\sigma_i} = C_i^{\sigma_i^-} \vee C_i^{\sigma_i^+}$, where $C_i^{\sigma_i^-}$ and $C_i^{\sigma_i^+}$ have no common literals.

\item[(ii)] $C_i^{\sigma_i^+}$ can be an empty clause itself, but $C_i^{\sigma_i^-}$ cannot be an empty itself clause.

\item[(iii)] $\bigwedge_{i=1}^{m}C_i^{\sigma_i^+}$ is a standard contradiction, that is, for any $(x_1,\cdots, x_m) \in \Pi_{i=1}^m C_i^{\sigma_i^-}$, there exists at least one complementary pair among $\{x_1,\cdots, x_m\}$.
\end{enumerate} 
The resulting clause $\bigvee_{i=1}^{m}C_i^{\sigma_i^+}$, denoted as $C_m^{s\sigma}(C_1,\cdots, C_m)$ (here “ $s$ ” means “standard”, $\sigma = \bigcup _{i=1}^m \sigma_i$), is called a standard contradiction separation clause (S-CSC) of $C_1,\cdots, C_m $, and $\bigwedge_{i=1}^{m}C_i^{\sigma_i^-}$ is called a separated standard contradiction (S-SC). 

\end{definition}

\begin{definition}\label{Def3.5}\cite{xu2018contradiction}
\textbf{(Quasi-Contradiction Separation Rule in First-Order Logic) }Suppose a clauses set $S = \{C_1, C_2,\cdots, C_m\}$ in first-order logic. Without loss of generality, assume that there does not exist the same variables among $C_1, C_2,\cdots, C_m$ (if there exists the same variables, there exists a rename substitution which makes them different). The following inference rule that produces a new clause from $S$ is called a quasi-contradiction separation rule, in short, a Q-CS rule:

For each $C_i (i =1, 2,\cdots, m)$, firstly applying a substitution $C_i (i =1, 2,\cdots, m)$ to $C_i$ ($\sigma_i$ could be an empty substitution but not necessary the most general unifier), denoted as $C_i^{\sigma_i}$. Then separate $C_i^{\sigma_i}$ into two sub-clauses $C_i^{\sigma_i^-}$ and $C_i^{\sigma_i^+}$ such that 
\begin{enumerate}
\item[(i)]  $C_i^{\sigma_i} = C_i^{\sigma_i^-} \vee C_i^{\sigma_i^+}$, where $C_i^{\sigma_i^-}$ and $C_i^{\sigma_i^+}$ have no common literals.

\item[(ii)] $C_i^{\sigma_i^+}$ can be an empty clause itself, but $C_i^{\sigma_i^-}$ cannot be an empty itself clause.

\item[(iii)] $\bigwedge_{i=1}^{m}C_i^{\sigma_i^+}$ is unsatisfiable.
\end{enumerate} 
The resulting clause $\bigvee_{i=1}^{m}C_i^{\sigma_i^+}$, denoted as $C_m^{q\sigma}(C_1,\cdots, C_m)$ (here “ $q$ ” means “quasi”, $\sigma = \bigcup _{i=1}^m \sigma_i$), is called a quasi-contradiction separation clause (Q-CSC) of $C_1, C_2$, $\cdots, C_m $, and $\bigwedge_{i=1}^{m}C_i^{\sigma_i^-}$ is called a separated quasi-contradiction (S-QC).
\end{definition}

\begin{definition}\label{Def3.6}\cite{xu2018contradiction} 
Suppose a clauses set $S = \{C_1, C_2,\cdots, C_m\}$ in first-order logic. $\phi_ 1, \phi_ 2,\cdots, \phi_ t$ is called a quasi-contradiction separation based deduction sequence (or a Q-CS based deduction sequence from $S$ to a clause $\phi_ t$. If $\phi_ i (i = 1, 2, \cdots, t)$ satisfies:  
\begin{enumerate}
\item[(1)] $\phi_ i \in  S $, or
\item[(2)] there exist $r_1, r_2,\cdots, r_{k_i} < i$, $\phi_ i = C^{s\sigma_i}_{r_{k_i}}(\phi_{r_1},\phi_{r_2},\cdots, \phi_{r_{k_i}})$, where $\sigma_ i = \bigcup _{j=1}^{k_i} \sigma_j$, $\sigma_ j$ is a substitution to $\phi_{r_j}$, $j =1,2,\cdots, k_i $.  
\end{enumerate}
\end{definition}

\begin{definition}\label{Def3.7} \cite{xu2018contradiction}
Suppose a clauses set $S = \{C_1, C_2,\cdots, C_m\}$ in first-order logic. $\phi_ 1, \phi_ 2,\cdots, \phi_ t$ is called a standard contradiction separation based dynamic deduction sequence (or a   S-CS based dynamic deduction sequence from $S$ to a clause $\phi_ t $. If $\phi_ i (i = 1, 2, \cdots, t)$ satisfies:
\begin{enumerate}
\item[(1)] $\phi_ i \in  S $, or
\item[(2)] there exist $r_1, r_2,\cdots, r_{k_i} < i$, $\phi_ i = C^{q\sigma_i}_{r_{k_i}}(\phi_{r_1},\phi_{r_2},\cdots, \phi_{r_{k_i}})$, where $\sigma_ i = \bigcup _{j=1}^{k_i} \sigma_j$, $\sigma_ j$ is a substitution to $\phi_{r_j}$, $j =1,2,\cdots, k_i $.  
\end{enumerate}
\end{definition}

\begin{theorem}\label{thm3.3}\cite{xu2018contradiction}
\textbf{(Soundness Theorem of the S-CS Based Dynamic Deduction in First-Order Logic)} Suppose a clauses set $S = \{C_1, C_2,\cdots, C_m\}$ in first-order logic. $\phi_ 1, \phi_ 2,\cdots, \phi_ t$ is an  S-CS based dynamic deduction from $S$ to a clause $\phi_ t $. If $\phi_ t$ is an empty clause, then $S$ is unsatisfiable.
\end{theorem}

\begin{theorem}\label{thm3.4}\cite{xu2018contradiction}
\textbf{(Completeness of the S-CS Based Dynamic Deduction in First-Order Logic)} Suppose a clauses set $S = \{C_1, C_2,\cdots, C_m\}$ in first-order logic. If $S$ is unsatisfiable, then there exists an S-CS based dynamic deduction from $S$ to an empty clause.
\end{theorem}

\begin{theorem}\label{thm3.5}\cite{xu2018contradiction}
\textbf{(Soundness of the Q-CS Based Dynamic Deduction in First-Order Logic)} Suppose a clauses set $S = \{C_1, C_2,\cdots, C_m\}$ in first-order logic. $\phi_ 1, \phi_ 2,\cdots, \phi_ t$ is an  Q-CS based dynamic deduction from $S$ to a clause $\phi_ t $. If $\phi_ t$ is an empty clause, then $S$ is unsatisfiable.
\end{theorem}

\section{\label{sec:level4}
Dynamic Automated Deduction by Contradiction Separation: The Standard Extension Method for Propositional Logic}

According to the dynamic deduction theory of contradiction separation, each deduction step involves dividing the participating clause set $\{C_1, C_2, \cdots, C_m\}$ into two components: a contradiction part and a contradiction-separation clause part.
Formally, each participating clause $C_i (i=1, 2, \cdots, m)$ is decomposed as $C_i = C_i^{-} \vee C_i^{+}$, where $C_i^{-}$ and $C_i^{+}$ denote, respectively, the negative and positive literal subsets of $C_i$. If the conjunction of all negative parts $\bigwedge_{i=1}^{m} C_i^{-}$ is unsatisfiable, then a contradiction separation is obtained, and the corresponding contradiction-separation clause is derived as $\bigvee_{i=1}^{m} C_i^{+}$. Thus, the efficiency of the dynamic deduction process depends critically on how contradictions are constructed. 

In the dynamic deduction method based on the Standard Extension (SE), the participating clauses and literals are extended through complementary literals, generating what is termed a standard contradiction. This mechanism provides an efficient and systematic means of realizing the contradiction separation theory in automated deduction.

\subsection{\label{sec:level4.1}ConceptS of the Standard Extension in Propositional Logic}
The core principle of the Standard Extension–based dynamic contradiction separation is to extend clauses through complementary literals to generate contradiction-separation clauses.
The intuitive idea of this contradiction-construction process in propositional logic is illustrated as follows.
\begin{enumerate}
\item[-] Select a clause in the clauses set $S = \{C_1, C_2,\cdots, C_n\}$ and denote it as $D_1$, and select a literal $x_1$ in $D_1$. Then extend the clause $D_1$ according to the literal $x_1$, and select a new clause which includes $\neg$$x_1$ in $S$ and denote it as $D_2$. 
\item[-] Select $x_2$ in $D_2 -\{\neg$$x_1\} $, and then proceed according to the literal $x_2$ to extend, that is, select a new clause which includes $\neg$$x_2$ in $S$ and denote it as $D_3$. 
\item[-] Select the literal $x_3$ in $D_3 -\{\neg$$x_1,  \neg$$x_2\} $, and then extend according to literal $x_3 $, that is, select a new clause includes $ \neg$$x_3$ in $S$ and denote it as $D_4 $. 
\item[-] Repeat this process in this way, the extension process will be terminated until a certain condition is met, and set the last extended literal at this time as $\neg$$x_{k-1}$, where the last extended clause is $D_k$.

\item[-] Connecting the extended literals $x_1 $, $x_2 $, $\cdots$, $x_{k-1} $, $\neg$$x_{k-1}$ as a boundary line, and the line divides each clause $D_i(1 \leq i \leq k)$ into two parts $D_i^+$ and $D_i^-$ as shown in Figure \ref{fig1}. Meanwhile, for each literal $x$ in $D_i(1 \leq i \leq k)$, if $x$ is complementary with $x_j$ in the extended literal set $\{x_1, x_2, \cdots, x_{i-1}\}$, then put $x_j$ in $D^-_i$. For intuitive representation, put $x$ on the same horizontal line with $x_j $.  In this case, the separated parts $D_i^-$ are taken together to form a contradiction based on standard extension of $D_1, D_2,\cdots, D_k$, i.e., $\bigwedge_{i=1}^kD_i^-$. The remaining parts form a new clause i.e., $\bigvee_{i=1}^kD_i^+$, which constitutes a contradiction separation clause based on standard extension of $D_1, D_2,\cdots, D_k $.   
 
\begin{figure}[hbt]\centering
  \includegraphics[width=0.55\textwidth]{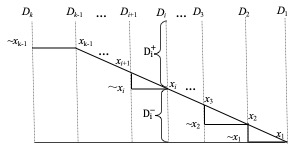} 
  \caption{Contradiction separation based on standard extension in propositional logic}
  \label{fig1}
\end{figure}
 
\end{enumerate}
\begin{definition}\label{Def4.1}
\textbf{(Contradiction and contradiction separation clause based on standard extension)} In propositional logic, let $S = \{C_1, C_2,\cdots, C_n\} $, $D_1, D_2,\cdots, D_m$ be involved clauses for constructing contradiction based on standard extension in $S$. If
\begin{enumerate}
\item[(1)] $D_i = D_i^+ \vee D_i^- (1 \leq i \leq m)$,
\item[(2)] $D_1^-= x_1$, $D_2^-= x_2 \vee \neg$$x_1$, $D_i^- =(x_i \vee  \neg$$x_{i-1}) \vee \bigvee_{x \in D^i}x$ $(3 \leq i \leq m -1)$, $D_m^- =  \neg$$x_{m-1} \bigvee _{x \in D^m} x $, where $D^j \subseteq \bigcup_{h=1}^{j-2}\{\neg$$x_h\}$ $(3 \leq j \leq m)$, 
\end{enumerate}
then $\bigwedge_{i=1}^mD_i^-$ is called a contradiction based on standard extension, and $R_s(D_1$, $D_2, \cdots, D_m) = \bigvee_{i=1}^mD_i^+$ is a contradiction separation clause based on standard extension of $D_1, D_2, \cdots, D_m$. The literal $x_i$ in $D_i^-$($ 1 \leq i \leq m -1 $) is called the extended literal, and its clause $D_i$ is called the extended clause. $\{x_1, x_2,\cdots, x_{m-1}\}$ are the literals on the main boundary, and $\{ \neg$$ x_1,  \neg$$x_2, \cdots,  \neg$$x_{m-1}\}$ are  the literals on the secondary boundary. 

\end{definition}
\begin{remark}\label{rmk4.1}

\begin{enumerate}
\item[(1)]  For the convenience of presentation, we call the process of the contradictions   construction and generation of contradictions separation clause based on standard extension in Definition \ref{Def4.1} is a contradiction separation based on standard extension.

\item[(2)] In the contradictions separation based on standard extension, the participating clauses can be repeated.

\item[(3)] In the contradictions separation based on standard extension, if there are only two clauses involved, then the contradictions separation is the same as binary resolution in this case. If multiple clauses participate in contradictions separation generally, then the contradictions are constructed by many clauses, and can be deleted. Therefore, the binary resolution is a special form of the contradictions separation based on the standard extension, and hence the contradictions separation based on the standard extension is an important extension of the binary resolution.

\item[(4)] In the contradictions separation based on standard extension, if the extended clauses are involved by all the clauses in the original clauses set $S$, and there exist some literals above the  main boundary, then $S$ is satisfiable, and the set of the extended literals in the main boundary is a satisfiable example of $S$.
\end{enumerate}
\end{remark}

\begin{example}\label{Exa1}
Let $S = \{C_1, C_2, C_3\}$  be the clauses set in propositional logic, $C_1$ = $\neg$$p \vee $$\neg$$ q \vee t$, $C_2$ = $p \vee t$, $C_3$ = $\neg$$ t $, where $p, q, t$ are the propositional variables. There exists a contradictions separation based on standard extension of $S$ as follows.

\begin{tabular}{lccc}\\\hline
  $C_1$ & $C_2$ & $C_3$\\
  \hline
  $\neg q$& & \\
  $\neg p$& $p$&\\
  $t$& $t$&$\neg t$\\\hline\\
\end{tabular}

Then $($$\neg$$ t) \wedge(p \vee  t) \wedge($$\neg$$ p \vee  t)$ is a contradiction  based on standard extension of $C_1, C_2, C_3$, the corresponding contradiction separation clause is $R_s(C_3, C_2, C_1)$ = $\neg$$ q $, where $\neg$$ t$ and $p$ are extended literals, $C_1, C_2, C_3$ are extended clauses.
\end{example}

\begin{theorem} 
In propositional logic, the contradiction $\bigwedge_{i=1}^mD_i^-$ in Definition \ref{Def4.1} is a standard contradiction.
\begin{proof}
It can be proved according to the definition of contradiction based on standard extension in Definition \ref{Def4.1} and Theorem 3.1.1 in \cite{xu2018contradiction}. 	
\end{proof}
\end{theorem}

\begin{remark}\label{rmk4.2}
In the process of contradictions separation based on standard extension, if the clauses involved are $C_1, C_2,\cdots, C_k$, and the contradictions separation clause is $R_s(C_1, C_2,\cdots, C_k)$, then $R_s(C_1, C_2,\cdots, C_k)$ may be different, or does not even exist when $C_1$, $C_2$, $\cdots$, $C_k$ participate in the contradictions separation in different order. Therefore, The result of the contradictions  separation clause based on standard extension depends on the order of extended clauses.
Furthermore, the result of $R_s(C_1, C_2,\cdots, C_k)$ is related to the extended literals that participate in the contradictions separation based on standard extension.
\end{remark}

\begin{example}\label{Exa2}
Suppose the clauses set $S = \{C_1, C_2, C_3, C_4\}$ in propositional logic, $C_1 =p $, $C_2 =$$\neg$$ p \vee $$\neg$$q \vee r $, $C_3 =r \vee q $, $C_4 =$$\neg$$r$, where $p, q, r$ are propositional variables. Then there exists a contradiction separation clause based on standard extension of $S$:

\begin{tabular}{lccc}\\\hline
  $C_4$& $C_3$& $C_2$& $C_1$\\ \hline
  & & $r$ & \\
  $\neg$$r$ & $r$&  & \\
  & $q$& $ \neg$$q$  & \\
  &  & $\neg$$p$   &  $p$\\\hline\\
\end{tabular}

Therefore, $R_s(C_1, C_2, C_3, C_4)=r$. If the order of the clauses are changed, then $R_s(C_4, C_3, C_2, C_1)= \emptyset$, $R_s(C_1, C_2, C_4, C_3)$ and $R_s(C_1, C_4, C_2, C_3)$ does not exist. If the selection of extended literal is different, $R_s(C_1, C_2, C_3, C_4)$ may not exist, concretely, if the extended literal in $C_2$ of $R_s(C_1, C_2, C_3, C_4)$ is selected as $r$, then $R_s(C_1, C_2, C_3, C_4)$ does not exist.
\end{example}

It can be seen from Example \ref{Exa2}, in the process of contradictions separation based on standard extension, the contradictions separation may be different even if each clause in the contradiction only includes the main boundary and the secondary boundary, and the order of the extended literal or the extended clauses is changed.

\begin{remark}\label{rmk4.3} 
The law of associativity in the separation of contradictions based on standard extension does not hold. A simple example is shown as follows.

Suppose the set of clauses $S = \{C_1, C_2, C_3, C_4\} $ in propositional logic, $C_1 =$$\neg$$ p$ $\vee $$\neg$$ q \vee t$, $C_2 = p \vee t$, $C_3 = q \vee t$, $C_4 = $$\neg$$t$, where $p, q, t$ are propositional variables, then 

\begin{tabular}{lcc}\\\hline
  $C_1$& $C_2$& $C_4$\\ \hline
  $\neg$$ q$& &   \\
  $\neg$$p$ & $p$&  \\
  $t$& $t$& $\neg$$t$ \\\hline\\ 
\end{tabular}

$C_5$ = $R_s (C_4, C_2, C_1)$ = $\neg$$q$ = $R_s (R_s (C_1, C_2), C_4)$. However, $R_s (C_1, R_s (C_2, C_4))$ = $t \vee $$\neg$$ q$. Therefore, $R_s (R_s (C_1, C_2), C_4)$ $\neq$ $R_s (C_1, R_s (C_2, C_4))$.
\end{remark}

\begin{definition}\label{Def4.2}
\textbf{(Refutation Sequence under the Standard Extension in Propositional Logic)}  
Let $S = \{C_1, C_2, \ldots, C_n\}$ be a set of clauses in propositional logic, and let  
$w = \{\phi_1, \phi_2, \ldots, \phi_t\}$ denote a dynamic deduction sequence of contradiction separation based on the Standard Extension (S-CS) from $S$ to $\phi_t$.  
Each $\phi_i \ (i = 1, 2, \ldots, t)$ in the sequence satisfies one of the following conditions:
\begin{enumerate}
    \item[(1)] $\phi_i \in S$, or
    \item[(2)] there exist indices $r_1, r_2, \ldots, r_{k_i} < i$ such that  
    $\phi_i = R_s(\phi_{r_1}, \phi_{r_2}, \ldots, \phi_{r_{k_i}})$,
    where $R_s$ denotes the inference rule of contradiction separation under the Standard Extension.
\end{enumerate}
If $\phi_t = \emptyset$, then $w$ is called a \textit{refutation sequence of contradiction separation based on the Standard Extension} for $S$.
\end{definition}

\begin{example}\label{Exa3}
Let the clauses set $S = \{C_1, C_2, C_3, C_4, C_5, C_6\}$ in propositional logic, $C_1$ = $x_1 \vee x_5$, $C_2$ = $\neg$$ x_1 \vee  x_2$, $C_3$ = $\neg$$ x_1\vee  $$\neg$$ x_2 \vee  x_3$, $C_4$ = $\neg$$ x_1 \vee  $$\neg$$ x_3 \vee  x_4$, $C_5$ = $\neg$$ x_1 \vee  $$\neg$$ x_2 \vee  $$\neg$$ x_4$, $C_6$ = $x_1 \vee $$\neg$$ x_5$, where $x_1, x_2, x_3, x_4, x_5$ are propositional variables. It is easy to validate $S$ is unsatisfiable. Using the method of contradictions separation based on standard extension, the following deduction can be obtained.

\begin{enumerate}
\item[(1)]  $C_1 = x_1 \vee x_5$
\item[(2)]  $C_2 =$$\neg$$ x_1 \vee x_2$
\item[(3)]  $C_3 =$$\neg$$ x_1 \vee $$\neg$$ x_2 \vee  x_3$
\item[(4)]  $C_4 =$$\neg$$ x_1 \vee $$\neg$$ x_3 \vee x_4$
\item[(5)]  $C_5 =$$\neg$$ x_1 \vee $$\neg$$ x_2 \vee $$\neg$$ x_4$
\item[(6)]  $C_6 = x_1 \vee $$\neg$$ x_5$
\\-----------------------
\item[(7)]  $C_7 = R_s (C_1, C_2, C_3, C_4, C_5) = x_5$
\item[(8)]  $C_8 = R_s (C_7, C_6, C_2, C_3, C_4, C_5)= \emptyset$
\end{enumerate}

Therefore, $w = \{C_1, C_2, C_3, C_4, C_5, C_6, C_7, C_8\}$ is a refutation of contradiction separation based on standard extension of $S$.
\end{example}

\subsection{\label{sec:level4.2}Soundness and Completeness of the Dynamic Deduction of the Contradictions Separation Based on Standard Extension}

In propositional logic, the dynamic deduction of contradiction separation based on the Standard Extension is a special case of the general contradiction separation framework.
Since the soundness of the general dynamic deduction has been established, the soundness of the Standard Extension–based dynamic deduction follows directly; that is, Theorem \ref{thm3.2} holds.

\begin{theorem}\label{thm3.2}
\textbf{(Soundness of the dynamic deduction of contradictions separation based on standard extension)} Let $S = \{C_1, C_2,\cdots, C_n\}$ be the clauses set in propositional logic, $\phi_ 1, \phi_ 2,\cdots, \phi_ t$ a deductive sequence of contradiction separation based on standard extension from $S$ to $\phi_t$. If $\phi_ t = \emptyset$, then $S$ is unsatisfiable.
\end{theorem}

\begin{theorem}\label{thm3.3}
\textbf{(Completeness of dynamic deduction based on contradiction separation based on standard extension)} Let $S = \{C_1,\cdots, C_n\}$ be the clauses set without tautologies and pure literal clauses in propositional logic. If $S$ is unsatisfiable, there exists a deduction of contradictions separation based on standard extension from $S$ to the empty clause $\emptyset$.
\end{theorem}
\begin{proof}
Let's prove this theorem by induction the number $|S| $ of clauses in $S$. If $|S| = 2$, then $S$ includes two clauses, i.e., $S = C_1 \wedge C_2$. Since $S$ is unsatisfiable, then any literals in $C_1$ and any literals in $C_2$ are complementary literals. Suppose $C_1 = x_1 \vee C_1^0 $, $C_2 =$$\neg$$ x_1 \vee C_2^0$, then for any literal $p$ in the $C_1^0$, $p$ is complementary to $\neg$$x_1$, hence  $p$ is the same as $x_1$, i.e., $C_1 = x_1$. In the same way, $C_2 =$$\neg$$ x_1$. Therefore, the conclusion holds for $|S| = 2$.

If $| S | = 3$, then $S$ includes $3$ clauses. Suppose $S = C_1 \wedge C_2 \wedge C_3 $, $C_1 = p_{11} \vee p_{12} \vee \cdots \vee p_{1n_1}$, $C_2 = p_{21} \vee p_{22} \vee \cdots \vee p_{2n_2}$, $C_3 = p_{31} \vee p_{32} \vee \cdots \vee p_{3n_3}$, where $p_{ij}$ is a literal. $C_1, C_2, C_3$ is sorted by the number of literals, i.e., $| C_1 | \leq | C_2 | \leq | C_3 | $. Since $S$ is unsatisfiable, then for any  literals $p_{1i} \in C_1$, $p_{2j} \in C_2$, $p_{3k} \in C_3$  ($1 \leq i \leq n_1, 1 \leq j \leq n_2, 1 \leq k \leq n_3$), there exists a complementary pair in  $\{p_{1i}, p_{2j}, p_{3k}\}$. Three cases exist as follows.
\begin{enumerate}
\item[(1)]  If $C_1$ is a unit clause, i.e., $C_1 = p_{11}$, then the following sub-cases exist.
\begin{enumerate}
\item[1)] If $C_2$ is a unit clause, then 
\begin{enumerate}
\item[i) ] if $C_2 = $$\neg$$ p_{11}$, then $C_3$ can be of any form.

\item[ii) ] if $C_2 = p_{21}$, $p_{21}$ and $p_{11}$ are not the same, then
\begin{enumerate}
\item [-] if $C_3$ is a unit clause, then $C_3 = $$\neg$$ p_{11} $, or $C_3 = $$\neg$$ p_{21} $.

\item[-] if $C_3$ is a 2-ary clause, $p_{11} \wedge p_{21}\wedge C_3 = p_{11} \wedge p_{21} \wedge (p_{31}\vee p_{32})= (p_{11} \wedge p_{21} \wedge p_{31})\vee (p_{11} \wedge p_{21}\wedge p_{32})$, then $p_{31}$, $p_{32}$ are complementary to $p_{11}$, $p_{21}$, i.e., $C_3 = $$\neg$$ p_{11} \vee $$\neg$$ p_{21}$.

\item[-] if $C_3$ includes more than or equal to $ 3$ literals, then $p_{11} \wedge p_{21}\wedge C_3 = p_{11} \wedge p_{21} \wedge (p_{31}\vee p_{32}\vee C_3^0)= (p_{11} \wedge p_{21} \wedge p_{31})\vee (p_{11} \wedge p_{21}\wedge p_{32})\vee (p_{11} \wedge p_{21}\wedge C_3^0)$. From the discussion above, $p_{31}$, $p_{32}$ are complementary to $p_{11}$, $p_{21}$. In this case, if $C_3^0 \neq \emptyset$, then for any literal $x$ in the $C_3^0$, $\{p_{11}, p_{21}, C_3^0\}$ includes no complementary pair. This is a contradictory to $\{p_{1i}, p_{2j}, p_{3k}\}$ includes a complementary pair.
\end{enumerate}
\end{enumerate}
\item[2)] If $C_2$ is a 2-ary clause, then the following cases exist.
\begin{enumerate}
\item[i) ] If $C_2$ includes $p_{11}$, then $C_1 \wedge C_2 \wedge C_3 = p_{11} \wedge (p_{11} \vee p_{21}) \wedge C_3 = p_{11} \wedge C_3$, and $C_3$ includes at least two literals, hence this is a contradictory to $\{p_{1i}, p_{2j}, p_{3k}\}$ includes a complementary pair.

\item[ii) ] If $C_2$ includes $ \neg$$p_{11}$, then $C_1 \wedge C_2 \wedge C_3 =  \neg$$p_{11} \wedge (p_{11} \vee p_{21}) \wedge C_3 = p_{11} \wedge p_{21} \wedge C_3$. $C_3$ includes at least two literals. According to the discussion in 1), $C_3 =  \neg$$p_{11} \vee  \neg$$p_{21}$.

\item[iii) ] If all the literals in $C_2$ are independent of $p_{11}$, then $C_1 \wedge C_2 \wedge C_3$ =  $\neg$$p_{11} \wedge (p_{21} \vee p_{22}) \wedge C_3$ = $(p_{11} \wedge p_{21} \wedge C_3) \vee (p_{11} \wedge p_{22} \wedge C_3)$. $C_3$ includes at least two literals. According to the discussion in 1), this is a contradictory to $\{p_{1i}, p_{2j}, p_{3k}\}$ includes a complementary pair.
\end{enumerate}
\item[3)] If $C_2$ includes more than or equal to 3 literals, then $C_2$ includes $\neg$$p_{11}$. According to the discussion in 2), $C_1 \wedge C_2 \wedge C_3 =  p_{11} \wedge (\neg$$p_{11} \vee p_{21}\vee p_{22} \vee C_2^0) \wedge C_3 = p_{11} \wedge (\neg$$p_{11} \vee p_{21}\vee p_{22} \vee C_2^0) \wedge C_3$. Therefore, this is a contradictory to $\{p_{1i}, p_{2j}, p_{3k}\}$ includes a complementary pair.
\end{enumerate}
\item[(2)] If $C_1$ includes more than or equal to 2 literals, similar to the discussion in (1), this is a contradictory to $\{p_{1i}, p_{2j}, p_{3k}\}$ includes a complementary pair.
\end{enumerate}

Therefore, if $S$ is unsatisfiable and $| S | = 3 $, then $S$ have 5 forms, i.e., $\{p_{11} \wedge  \neg$$p_{11} \wedge C_3, p_{11} \wedge p_{21} \wedge  \neg$$p_{11}, p_{11} \wedge p_{21} \wedge  \neg$$p_{21}, p_{11} \wedge p_{21} \wedge ( \neg$$p_{11} \vee  \neg$$p_{21}), p_{11} \wedge ( \neg$$p_{11} \vee p_{21})\wedge ( \neg$$p_{11} \vee  \neg$$p_{21})\} $. Obviously, a refutation of contradiction separation based on standard extension of $S$ can be obtained.

Suppose the conclusion holds for $| S | \leq n$. We will prove that the conclusion holds for $| S | = n+1 $. 

Let $S = \{C_1, C_2,\cdots, C_{n+1}\}$ be a unsatisfiable clauses set, and $C_{n+1} =\bigvee_{t=1}^Tx_t $, then $C_1 \wedge C_2 \wedge \cdots \wedge C_{n+1}$ is unsatisfiable if and only if for any $t \in  \{1,\cdots, T\}$, $C_1 \wedge C_2 \wedge  \cdots \wedge  C_n \wedge x_t$ is unsatisfiable. Then for any $t \in  \{1,\cdots, T\}$, suppose there exist complementary literals of $x_t$ in $C_1, C_2,\cdots, C_n $. Without loss of generality, let all the clauses included $\neg$$x_t$ in $C_1, C_2, \cdots, C_{n}$ be $C_i =  \neg$$x_t \vee C_i^0$ $(i = 1,2, \cdots, k_t)$. Thus 
\begin{equation}
\begin{aligned}
&C_1 \wedge C_2 \wedge \cdots \wedge C_n \wedge x_t\\ 
&= ( \neg x_t \vee (C_1^0 \wedge \cdots \wedge C_{k_t}^0))\wedge C_{k_t+1} \wedge \cdots \wedge C_n \wedge x_t. \\
&= ( \neg x_t \wedge C_{k_t}^0\wedge \cdots \wedge  C_n \wedge x_t) \vee (C_1^0 \wedge \cdots \wedge C_{k_t}^0  \wedge C_{k_t+1} \wedge \cdots \wedge C_n \wedge x_t)\\
&=C_1^0 \wedge \cdots \wedge C_{k_t}^0  \wedge C_{k_t+1}^0 \wedge \cdots \wedge C_n \wedge x_t.\nonumber
\end{aligned}
\end{equation}

Therefore, $C_1 \wedge C_2 \wedge \cdots \wedge  C_n \wedge x_t$ is  unsatisfiable if and only if $C_1^0 \wedge \cdots \wedge C_{k_t}^0  \wedge C_{k_t+1} \wedge \cdots \wedge C_n \wedge x_t$ is unsatisfiable. Note that $\{C_1^0, \cdots, C_{k_t}^0, C_{k_t+1},\cdots, C_n\}$ does not include $ \neg$$x_t$, hence $C_1 \wedge C_2 \wedge \cdots \wedge  C_n \wedge x_t$ is unsatisfiable if and only if $C_1^0 \wedge \cdots \wedge C_{k_t}^0  \wedge C_{k_t+1} \wedge \cdots \wedge C_n$ is unsatisfiable.

Since $C_1^0 \wedge \cdots \wedge C_{k_t}^0  \wedge C_{k_t+1} \wedge \cdots \wedge C_n$ is  unsatisfiable, and $| C_1^0 \wedge \cdots \wedge C_{k_t}^0  \wedge C_{k_t+1} \wedge \cdots \wedge C_n| \leq n$, according to the inductive hypothesis, there exists a refutation of contradiction separation based on standard extension of the clauses set $\{C_1^0, \cdots, C_{k_t}^0, C_{k_t+1}, \cdots, C_n\}$,  i.e., $w_t^* = \{D_{t_1}^*, \cdots, D_{t_u}^* = \emptyset\}$.

For the refutation $w_t^*$, the following cases exist.
\begin{enumerate}
\item[(1)] If all the involved clauses in $w_t^*$ are from $C_{k_t+1}, \cdots, C_n$, then $w_t^*$ is also the refutation of contradiction separation based on standard extension of $S$. In this case, this theorem holds. 
\item[(2)] If the involved clauses in $w_t^*$ include $C_j^0 (1 \leq j \leq k_t)$, since $C_j$ = $\neg$$x_t \vee C_j^0$ $(i = 1,2, \cdots, k_t)$, then add the literal $\neg$$x_t$ to $C_j^0$. Therefore, all the involved clauses $C_j^0 (1 \leq j \leq k_t)$ in the refutation $w_t^*$ are restored to $C_j (1 \leq j \leq k_t)$. In this case, there exists a deduction of the separation of contradictions based on standard extension from $C_1\wedge\cdots \wedge C_{n}$ to $\neg$$x_t$, hence there exists a refutation of contradiction separation based on standard extension of $C_1\wedge\cdots\wedge C_{n}\wedge x_t$. Therefore, this theorem holds for any $t \in \{1, \cdots, T\} $. Since  $\{C_1^0, \cdots, C_{k_t}^0, C_{k_t+1}, \cdots, C_n\}$ does not include $ \neg$$x_t$, then the following subcases exist.
\begin{enumerate}
\item[1)] $ \neg$$x_t$ is the literal in contradiction separation based on standard extension of $C_1, \cdots,C_{k_t}, C_{k_t+1},\cdots, C_n$.
\item[2)] Since $\{C_1^0, \cdots, C_{k_t}^0, C_{k_t+1}, \cdots, C_n\}$ does not include $\neg$$x_t$, then the extended literals in contradiction separation based on standard extension do not include $x_t$ and $ \neg$$x_t$. Therefore, the remaining $ \neg$$x_t$ does not affect contradiction separation based on standard extension.
\end{enumerate}
\end{enumerate}

Denote the contradiction separation deduction based on standard extension from $\{C_1, \cdots,C_{k_t}, C_{k_t+1},\cdots, C_n\}$ to $ \neg$$x_t$ as $w_t^0 = \{D_{t_1}^0, \cdots, D_{t_u}^0$ = $\neg$$x_t\}$, where $t \in  \{1,\cdots,T\}$. Note that the deduction of contradictions separation based on standard extension can be constructed as follows.
\begin{equation}
\begin{aligned}
& R_s ( \neg x_2, C_{n+1},  \neg x_1) = \bigvee _{i=3}^Tx_i.\\
& R_s ( \neg x_4, R_s( \neg x_2, C_{n+1},  \neg x_1),  \neg x_{3}) = R_s( \neg x_4, \bigvee _{i=3}^Tx_i,  \neg x_3) =  \bigvee _{i=5}^Tx_i.\\
&\cdots\nonumber
\end{aligned}
\end{equation}

Two cases exist as follows.
\begin{enumerate}
\item[(1)] If $T$ = $2m$ is an even number, then $R_s ( \neg$$x_{2m}, R_s( \neg$$x_{2m-2}, R_s(\cdots),  \neg$$x_{2m-3})$,  $\neg$$x_{2m-1})$ = $\emptyset $.
\item[(2)] If $T=2m+1$ is an odd number, then
$R_s ( \neg$$x_{2m+1}, R_s( \neg$$x_{2m}, R_s(\cdots),  \neg$$x_{2m-1}),$  $\neg$$x_{2m-1})$ = $\emptyset $.
\end{enumerate}

Hence a refutation of contradiction separation based on standard extension of $\{C_1, \cdots,C_{k_t}, C_{k_t+1},\cdots, C_n, C_{n+1}\}$ can be obtained. Therefore, the conclusion holds for $| S | = n+1 $.
\end{proof}

\begin{remark}\label{rmk3.5}
The completeness of the contradiction separation deduction based on the Standard Extension follows directly from that of binary resolution.  
Since binary resolution can be regarded as a special case of the Standard Extension–based contradiction separation, the completeness of the latter is guaranteed by the well-established completeness of the former.
\end{remark}

\subsection{\label{sec:level4.3}A Unified Dynamic Deduction Method Based on the Standard Extension in Propositional Logic}

This section presents a unified realization of the contradiction separation framework under the Standard Extension, which provides a systematic procedure for determining the satisfiability or unsatisfiability of propositional formulas.  
In particular, when a formula is satisfiable, the method is capable of constructing a corresponding satisfying assignment.

The unified dynamic deduction process for contradiction separation based on the Standard Extension in propositional logic is described as follows.

\begin{enumerate}
\item[Step 0.] (Pre-treatment) Let $S_0 = \{C_1, C_2, \cdots, C_n\}$ be a set of clauses. Use pure literal rule and tautology rule in $S_0$, and get a new clauses set $S$.

\item[Step 1.] For the clauses set $S$, select a clause in $S$, and denote it as $D_1 $. Then select a literal in $D_1$, and denote it as $x_1$. Extend the literal $x_1$, i.e., select its complementary literal $ \neg$$x_1$ in $S$, and denote its clause as $D_2(x_1)$. In this case, the clause $D_1$ is divided into two parts: clause $D_1^-$ and clause $D_1^+$, where $D_1^-= x_1$, $D_1^+=D_1-D_1^-$, which means that $D_1^+$ is gotten by deleting $D_1^-$ from $D_1$.
\item[Step 2.] Select a literal in $D_2(x_1) - \{\neg$$x_1\}$, and denote it as $x_2$. Extend the literal $x_2$, i.e., select its complementary literal $ \neg$$x_2$ in $S$, and denote its clause as $D_3(x_2)$. In this case, the clause $D_2(x_1)$ is divided into two parts: clause $D_2^-(x_1)$ and clause $D_2^+(x_1)$, where $D_2^-(x_1)= x_2 \vee  \neg$$x_1$, $D_2^+(x_1)= D_2(x_1) -D_2^-(x_1)$, which means that $D_2^+(x_1)$ is gotten by deleting $D_2^-(x_1)$ from $D_2(x_1)$.

\item[$\cdots$ ]

\item[Step $i (i \geq 3)$.] Select a literal in $D_{i}(x_{i-1})- \{\neg$$x_{i-1}\} -\cdots - \{\neg$$x_1\}$, and denote it as $x_i$. Extend the literal $x_i$, i.e., select its complementary literal $ \neg$$x_i$ in $S$, and denote its clause as $D_{i+1}(x_i)$. In this case, the clause $D_{i+1}(x_i)$ is divided into two parts: clause $D_{i+1}^-(x_i)$ and clause $D_{i+1}^+(x_i)$, where $D_i^-(x_{i-1})=(x_i \vee  \neg$$x_{i-1}) \vee D_i^0, D_i^0 = \bigvee_{x \in D^i}x$, $D^i \subseteq \{ \neg$$x_j | j = 1,2,\cdots, i-2\}, D_i^+(x_{i-1})=D_i(x_{i-1}) - D_{i}^-(x_{i-1})$, which means that $D_{i+1}^+(x_i)$ is gotten by deleting $D_{i+1}^-(x_i)$ from $D_{i+1}(x_i)$.

\item[$\cdots$ ] 

\item[Step $k-1$.] Select a literal in $D_{k-1}(x_{k-2})- \{\neg$$x_{k-2}\} -\cdots - \{\neg$$x_1\}$, and denote it as $x_{k-1}$. Extend the literal $x_{k-1}$, i.e., select its complementary literal $ \neg$$x_{k-1}$ in $S$, and denote its clause as $D_{k}(x_{k-1})$. In this case, the clause $D_{k}(x_{k-1})$ is divided into two parts: clause $D_{k}^-(x_{k-1})$ and clause $D_{k}^+(x_{k-1})$, where $D_{k-1}^-(x_{k-2})=(x_{k-1} \vee  \neg$$x_{k-2}) \vee D_{k-1}^0$, $D_{k-1}^0 = \bigvee_{x \in D^{k-1}}x$, $D^{k-1} \subseteq \{ \neg$$x_j | j = 1,2,\cdots, k-3\}$, $D_{k-1}^+(x_{k-2})=D_{k-1}(x_{k-2}) - D_{k-1}^-(x_{k-2})$, which means that $D_{k}^+(x_{k-1})$ is gotten by deleting $D_{k}^-(x_{k-1})$ from $D_{k}(x_{k-1})$.

\item[Step $k$.] For the clause $D_k(x_{k-1})$, if it cannot be extended or no longer needs to be extended, then stop the extension. In this case, the clause $D_k(x_{k-1})$ is divided into two parts: the clause  $D_k^-(x_{k-1})$ and clause $D_k^+(x_{k-1})$, where $D_{k-1}^-(x_{k-2})$ =  $\neg$$x_{k-1} \vee D_{k}^0$, $D_{k}^0$ = $\bigvee_{x \in D^{k}}x$, and $D^{k} \subseteq \{ \neg$$x_j | j = 1,2,\cdots, k-2\}$, $D_{k}^+(x_{k-1})$ = $D_{k}(x_{k-1}) - D_{k}^-(x_{k-1}) $.

\end{enumerate}

The contradiction based on the standard extension $\bigwedge_{i=1}^kD_i^-(x_{i-1})$ and the contradiction separation clause based on standard extension $R_s(D_1,D_2,\cdots,D_k)$ = $\bigvee_{i=1}^kD_i^+(x_{i-1})$  are obtained respectively, where $D_1^-(x_0)= D_1^-$, $D_1^+(x_0)= D_1^+$. Then the following cases exist.
\begin{enumerate}
\item[(1)] If $R_s(D_1,D_2,\cdots,D_k) = \emptyset$, then $S$ is  unsatisfiable.

\item[(2)]  If $R_s(D_1,D_2,\cdots,D_k) \neq \emptyset$, then the subcases exist.
\begin{enumerate}
\item[(i)] If some clauses of $S$ are not in the involved clauses $\{D_1,D_2,\cdots,D_k\}$, then $S = S \cup \{R_s(D_1,D_2,\cdots,D_k)\}$, and go to step 0.

\item[(ii)] Otherwise, if there exist the literal $y$ and index $j_0$, such that $y \in D_{j_0}^+ $, and  $x_{k-1},\cdots, x_{j_0}$ are different from $y $, and all the literals in $\{x_{k-1},\cdots, x_{j_0}\}$ are not the same with the literals in $\{x_{j_0-1},\cdots, x_1\}$,  then ($ \neg$$x_{k-1},\cdots,$ $\neg$$x _{j_0}, y, x_{j_0-1},\cdots, x_1$) is a satisfiable example of $S$, where $ \neg$$x_{k-1},\cdots,$ $ \neg$$x _{j_0}, y, x_{j_0-1},\cdots, x_1$ are the literals in $D_k, \cdots$, $D _{j_0+1}$, $D_{j_0}, D_{j_0-1},\cdots, D_1$, respectively.
\end{enumerate}

\end{enumerate}

\begin{remark}\label{rmk4.4}

\begin{enumerate}
\item[(1)] To further enhance the efficiency of the unified contradiction separation method based on the Standard Extension, it is useful to prioritize clauses and literals according to their characteristics, such as the number of logical symbols, the presence of ground terms, or assigned weights within literals and clauses.  
After sorting, the first literal in the ordered list is selected as the extended literal, and its complementary literal in the first clause is chosen as the corresponding extended clause.  
This process constructs the associated contradictions and their separations.  
Since sorting merely defines a priority order for selecting literals and clauses, it does not affect the completeness of the method.

\item[(2)] Additional constraints can be introduced to further improve efficiency when applying the unified method:
\begin{enumerate}
    \item All atoms of literals on the main boundary must be distinct; in other words, the literals should be neither identical nor complementary.
    \item When selecting literals on the main boundary, a literal should not duplicate any remaining literal used during the extension process.
    \item During the extension process, no complementary literal pairs should appear among the remaining literals.
\end{enumerate}
\end{enumerate}
\end{remark}

According to these principles, a unified algorithm for contradiction separation based on the Standard Extension in propositional logic is presented in Algorithm~\ref{agm1}.

\begin{algorithm}[h]
   \begin{scriptsize}
   \caption{A unified algorithm for deduction of contradictions separation based on standard extension in propositional logic} \label{agm1}
   \BlankLine
   \KwIn{$S= \{C_1, C_2, \cdots, C_n\}$ is the clauses set in propositional logic.}
   \KwOut{Judge the unsatisfiability and satisfiability of $S$. If $S$ is satisfiable, then also give its  counterexample.}
   Initiate: $T_0$ is the set of the conditions for stopping the contradiction extension.\\
   \While {$P_s \neq 0$ $\&$ $P_s \neq 1$ }{
   $S_0 \leftarrow S$ by pure literal rule and tautology rule;\\
   Select $D_1$ is $S_0$, and $x_1$ in $D_1$;\\
   Extend $x_1$ and select $\neg$$x_1$ in $S_0$, denote its clause as $D_2(x_1)$;\\
   k=1;\\
   $D_1 = D_1^- \cup D_1^+$, $D_1^- = x_1$, $D_1^+ = D_1 - D_1^-$.\\
     \While {$T_0$ is not met}{
     $k=k+1$;\\
     Select $x_k$ in $D_k(x_{k-1})-\bigvee_{j=1}^{k-1}\{$$\neg x_j\}$ in $S_0$, denote its clause as $D_{k+1}(x_k)$;\\
     $D_k = D_k^- \cup D_k^+$, $D_k^- =(x_i\vee \neg x_{i-1}) \vee D_k^0$, $D_k^0=\bigvee_{x \in D^k} x$, $D^k \subseteq \{\neg x_j | j=1,\cdots, k-2\}$.
     }
     $R_s=\bigvee_{i=1}^kD_i^+(x_{i-1})$.\\
     \If {$R_s$ = $\emptyset$}{
     $S$ is unsatisfiable, and $P_s=0$.\\
     }
     \Else {
     Denote the involved clause in $S$ as $IC$.\\
       \If  {$IC \neq  S$}{
           $S=S \cup R_s$.\\
           }
        \Else {
           $S$ is satisfiable, $P_s=1$;\\
           Choose $y \in D_{j_0}^+$ such that $y$ is different from the literals in $x_{k-1},\cdots, x_{j_0}$, then ($ \neg$$x_{k-1},\cdots,$ $\neg$$x _{j_0}, y, x_{j_0-1},\cdots, x_1$) is a satisfiable example of $S$.
           }
     }
   }
   \end{scriptsize}
\end{algorithm}

This algorithm systematically realizes the dynamic deduction framework of contradiction separation under the Standard Extension in propositional logic.  
By iteratively extending complementary literals and generating contradiction-separation clauses, it provides a constructive procedure for determining both satisfiability and unsatisfiability.  
If a contradiction is detected, the method yields a refutation demonstrating unsatisfiability; otherwise, it produces a concrete satisfying assignment, thereby ensuring both soundness and completeness.  
This unified process not only generalizes traditional resolution-based strategies but also serves as the operational foundation for subsequent extensions to first-order logic and for implementations in systems such as CSE, CSE\_E, CSI\_E, and CSI\_Enig as detailed in experimention section.

 \subsection{\label{sec:level4.4}Relationship Between the Standard Extension and Linear Resolution in Propositional Logic}

This section analyzes the theoretical connection between the Standard Extension–based contradiction separation and the classical linear resolution method.  
Both approaches are grounded in the resolution principle but differ in their structural formulation and inference dynamics.  
While linear resolution performs binary clause resolution in a sequential chain, the Standard Extension generalizes this process by allowing multi-clause cooperation through dynamic contradiction construction.  

Consequently, linear resolution can be regarded as a special case of the Standard Extension, where the inference is restricted to a single resolution path.

\begin{theorem}\label{thm4.6}
In propositional logic, let $S = \{C_1, C_2,\cdots, C_n\}$ be a clauses set, $D_i \subseteq C_i (1\leq i\leq n)$, $R_s(D_1, D_2, \cdots, D_n)$ be contradiction separation based on standard extension of $D_1, D_2, \cdots, D_n$, then there exists a linear resolvent $R_{n,n-1,\cdots,1} (R_{n,n-1,\cdots,2}$ $ (\cdots, R_{n,n-1} (D_n, D_{n-1}), D_{n-2}),\cdots), D_1) $ of $D_1, D_2, \cdots, D_n$, such that  
 $R_s(D_1, D_2,\cdots, D_n)$ = $R_{n,n-1,\cdots,1}(R_{n,n-1,\cdots,2}(\cdots, R_{n,n-1}(D_n, D_{n-1}), D_{n-2}),\cdots), D_1)$   
\end{theorem}

\begin{proof}
Suppose that the clauses participating in contradiction separation based on standard extension are $D_1, D_2,\cdots,D_k$, and its extended literals are $x_1, x_2,\cdots, x_{k-1}$. According to the definition of contradiction separation based on standard extension, the contradiction separation can be written as another form, i.e., $R_s (D_1,D_2,\cdots, $ $D_k)$ = $(D_1 - \{x_1\}) \vee (D_k(x_{k-1}) -\bigvee_{h=1}^{k-1}\{ \neg$$x_h\}) \vee \bigvee_{j=2}^{k-1}(D_j(x_{j-1})-\{x_j\}-\bigvee_{h=1}^{j-1}\{ \neg$$x_h\})$, where $k \geq 2 $, and $\bigvee_{j=2}^1(D_j(x_{j-1})-\{x_j\}-\bigvee_{h=1}^{j-1}\{ \neg$$x_h\})=\emptyset$. We need to prove that the linear resolvent is equal to contradiction separation.
 
We will prove this theorem by induction on the number $k$ of the involved clauses. If $k=2$, $R_s (D_1,D_2)$= $(D_1 -\{x_1\}) \vee (D_2(x_1)-\{ \neg$$x_1\})$. In this case, it is obvious that $R_s (D_1, D_2)$ = $R(D_2,D_1)$, that is, the conclusion holds for $k=2$.

Assuming that the conclusion holds for $k = n$. We need to prove that the conclusion holds for $k = n +1$. For the contradictions separation based on standard extension,
$R_s (D_1,D_2,\cdots,D_{n+1})$ = $(D_1 - \{x_1\}) \vee (D_k(x_{n+1}))-\bigvee_{h=1}^{n}\{ \neg$$x_h\}) \vee \bigvee_{j=2}^{n}(D_j(x_{j-1})-\{x_j\}-\bigvee_{h=1}^{j-1}\{\neg$$x_h\})$. According to the structure of linear resolution deduction and induction hypothesis, $R_{n+1,n,\cdots,1} (R_{n+1,n,\cdots,2} (\cdots (R_{n+1,n} (D_{n+1},D_n), \cdots),D_1)$ = $R_{n+1,n,\cdots,1} ((D_2 - \{x_2\}) \vee (D_{n+1}(x_{n}))-\bigvee_{h=2}^{n}\{ \neg$$x_h\}) \vee \bigvee_{j=3}^{n}(D_j(x_{j-1})-\{x_j\}-\bigvee_{h=2}^{j-1}\{ \neg$$x_h\}), D_1)$ = $R((D_2 - \{x_2\}) \vee (D_{n+1}(x_{n}))-\bigvee_{h=2}^{n}\{ \neg$$x_h\}) \vee \bigvee_{j=3}^{n}(D_j(x_{j-1})-\{x_j\}-\bigvee_{h=2}^{j-1}\{ \neg$$x_h\}), D_1)$.

Furthermore, according to construction method of contradiction separation based on standard extension, there exists $ \neg$$x_1$ in ($D_2 -\{x_2\}$), and its complementary literal $x_1 \in  D_1$. In the linear resolution deduction, if $(D_{n+1}(x_{n}))-\bigvee_{h=2}^{n}\{ \neg$$x_h\}) \vee \bigvee_{j=3}^{n}(D_j(x_{j-1})-\{x_j\}-\bigvee_{h=2}^{j-1}\{ \neg$$x_h\})$ includes $ \neg$$x_1$, then delete $ \neg$$x_1$ at the same time. Therefore, 

\begin{footnotesize}
$R((D_2 - \{x_2\}) \vee (D_{n+1}(x_{n}))-\bigvee_{h=2}^{n}\{ \neg$$x_h\}) \vee \bigvee_{j=3}^{n}(D_j(x_{j-1})-\{x_j\}-\bigvee_{h=2}^{j-1}\{ \neg$$x_h\}), D_1)$

= $((D_2 - \{x_2\}) \vee (D_{n+1}(x_{n}))-\bigvee_{h=2}^{n}\{ \neg$$x_h\}) \vee \bigvee_{j=3}^{n}(D_j(x_{j-1})-\{x_j\}-\bigvee_{h=2}^{j-1}\{ \neg$$x_h\})-\{ \neg$$x_1\}) \vee (D_1 -\{x_1\})$

= $(D_2 - \{x_2\}-\{ \neg$$x_1\}) \vee (D_{n+1}(x_{n}))-\bigvee_{h=2}^{n}\{ \neg$$x_h\}) \vee \bigvee_{j=3}^{n}(D_j(x_{j-1})-\{x_j\}-\bigvee_{h=1}^{j-1}\{ \neg$$x_h\}) \vee (D_1 -\{x_1\})$

= $(D_1 -\{x_1\}) \vee (D_{n+1}(x_{n}))-\bigvee_{h=1}^{n}\{ \neg$$x_h\}) \vee \bigvee_{j=3}^{n}(D_j(x_{j-1})-\{x_j\}-\bigvee_{h=1}^{j-1}\{ \neg$$x_h\})$
\end{footnotesize}

Therefore, $R_s (D_1,D_2,\cdots,D_{n+1})$ = $R_{n+1,n,\cdots,1} (R _{n+1,n,\cdots,2} (\cdots (R_{n+1,n }(D_{n+1},D_n), \cdots), D_1)$, that is, the conclusion holds $k$ = $n +1$.

\end{proof}

\begin{remark}\label{rmk4.6}
The inverse proposition of Theorem \ref{thm4.6} does not necessarily hold, that is, if $R_{1,2,\cdots,n} (R_{1,2,\cdots,n-1} (\cdots, R_{1,2 }(C_1, C_2), C_3), \cdots), C_n)$ is a linear resolvent of the clauses set $S = \{C_1, C_2,\cdots, C_n\}$, but it does not necessarily exist a contradictions separation based on standard extension in the order of $C_n, C_{n-1},\cdots, C_1$. A counterexample is shown as follows.

Let $C_1$ = $a \vee b \vee c \vee d$, $C_2$ = $\neg$$a $, $C_3$ = $\neg$$b$, $C_4$ = $\neg$$c$ in the propositional logic, where $a, b, c, d$ are the propositional variables. Then there exists a linear resolution deduction $R_{1,2,3,4}(R_{1,2,3}(R_{1,2}(C_1, C_2), C_3), C_4) = d$. Obviously, there exists no contradictions separation of based on standard extension in the order of $C_4, C_3, C_2, C_1$, such that $R_s (C_4, C_3, C_2, C_1) = d $.
\end{remark}

This theorem shows that the Standard Extension–based contradiction separation can be expressed as an equivalent sequence of linear resolution steps.  
Hence, the Standard Extension generalizes linear resolution by allowing the simultaneous participation of multiple clauses in a single deduction step, while retaining logical equivalence with the corresponding linear resolvent chain.  This relationship establishes a formal bridge between the classical resolution framework and the dynamic contradiction separation mechanism.

\section{\label{sec:level5}Dynamic Deduction by Contradiction Separation Based on the Standard Extension in First-Order Logic}

This section extends the framework of contradiction separation based on the Standard Extension (SE) from propositional logic to first-order logic.  
The goal is to establish a unified and logically sound deduction mechanism capable of handling quantified formulas, variable substitutions, and unification while preserving the theoretical properties of soundness and completeness established in the propositional case.

\subsection{\label{sec:level5.1}Concept of the Standard Extension in First-Order Logic}

Analogous to the propositional case, the Standard Extension in first-order logic constructs contradictions by extending literals and clauses through complementary pairs derived from substitution and unification.  
The underlying idea is to generalize the contradiction separation process so that it operates over quantified formulas with variable instantiation.  
The intuitive structure of the contradiction construction method based on the Standard Extension in first-order logic is illustrated as follows.

\begin{enumerate}
\item[-] Suppose that $ S = \{C_1, C_2,\cdots, C_n\}$ is a clauses set, all the clauses in  $ S $ are without any common variables. First select a clause in $S$ denoted as $D_1$, and select $x_1$ in $D_1$. According to the literal $x_1$ to extend, select the literal $y(x_1)$ in $S$ denoted its clause as $D_2(x_1)$, such that $x_1^{\theta(x_1)}$ =  $\neg$$y(x_1)^{\theta(y(x_1))}$, where $\theta(x_1)$ and $\theta(y(x_1))$ are substitutions of $x_1$ and $y(x_1)$ respectively, and merge the same literals in $D_1$ and $D_2(x_1)$, and obtain their instances $D_1^{\theta(x_1)}$ and $D_2(x_1)^{\theta(y(x_1))}$ respectively.

\item[-] Select a literal in $D_2(x_1)^{\theta(y(x_1))}- $$y(x_1)^{\theta(y(x_1))}$ denoted it as $x_2$. According to the literal $x_2$ to extend, select the literal $y(x_2)$ in $S$ denoted its clause as $D_3(x_2)$, such that $x_2^{\theta(x_2)}$ =  $\neg$$y(x_2)^{\theta(y(x_2))}$, where $\theta(x_2)$ and $\theta(y(x_2))$ are substitutions of $x_2$ and $y(x_2)$ respectively, and merge the same literals in $D_2(x_1)^{\theta(y(x_1))}$ and $D_3(x_2)$, and obtain their instances $D_2(x_1)^{\theta(y(x_1))\theta(x_2)}$ and $D_3(x_2)^{\theta(y(x_2))}$ respectively. Meanwhile, substitute $D_1^{\theta(x_1)}$ with $\theta(x_2)$, merge the same literals and get its instance $D_1^{\theta(x_1)\theta(x_2)}$. In this case, three extended clauses are gotten, i.e., $D_3(x_2)^{\theta(y(x_2))}$, $D_2(x_1)^{\theta(y(x_1))\theta(x_2)}$, and $D_1^{\theta(x_1)\theta(x_2)}$.

\item[-] Select a literal in $D_3(x_2)^{\theta(y(x_2))}- y(x_2)^{\theta(y(x_2))}-y(x_1)^{\theta(y(x_1))\theta(x_2)}$ denoted it as $x_3$. According to the literal $x_3$ to extend, select the literal $y(x_3)$ in $S$ denoted its clause as $D_4(x_3)$, such that $x_3^{\theta(x_3)}$ =  $\neg$$y(x_3)^{\theta(y(x_3))}$, where $\theta(x_3)$ and $\theta(y(x_3))$ are substitutions of $x_3$ and $y(x_3)$ respectively, and merge the same literals in $D_3(x_2)^{\theta(y(x_2))}$ and $D_4(x_3)$, and obtain their instances $D_3(x_2)^{\theta(y(x_2))\theta(x_3)}$ and $D_4(x_3)^{\theta(y(x_3))}$ respectively. Meanwhile, substitute $D_2(x_1)^{\theta(y(x_1))\theta(x_2)}$ and $D_1^{\theta(x_1)\theta(x_2)}$ with $\theta(x_3)$, merge the same literals and get their instances $D_2(x_1)^{\theta(y(x_1))\theta(x_2)\theta(x_3)}$ and $D_1^{\theta(x_1)\theta(x_2)\theta(x_3)}$. In this case, four extended clauses are gotten, i.e., $D_4(x_3)^{\theta(y(x_3))}$, $D_3(x_2)^{\theta(y(x_2))\theta(x_3)}$, $D_2(x_1)^{\theta(y(x_1))\theta(x_2)\theta(x_3)}$ and $D_1^{\theta(x_1)\theta(x_2)\theta(x_3)}$.

\item[-] Repeat the extended process above until some certain conditions are satisfied. Suppose the last extended literal is $x_{k-1}$ denoted its clause as $D_{k-1}(x_{k-2})$. According to the literal $x_{k-1}$ to extend, select the literal $y(x_{k-1})$ in $S$ denoted its clause as $D_k(x_{k-1})$, such that $x_{k-1}^{\theta(x_{k-1})}$ =  $\neg$$y(x_{k-1})^{\theta(y(x_{k-1}))}$, where $\theta(x_{k-1})$ and $\theta(y(x_{k-1}))$ are substitutions of $x_{k-1}$ and $y(x_{k-1})$ respectively, and merge the same literals in $D_{k-1}(x_{k-2})^{\theta(y(x_{k-2}))}$ and $D_{k}(x_{k-1})$, and obtain their instances $D_{k-1}(x_{k-2})^{\theta(y(x_{k-2}))\theta(x_{k-1})}$ and $D_{k}(x_{k-1})^{\theta(y(x_{k-1}))}$ respectively. Meanwhile, substitute $D_{k-2}(x_{k-3})^{\theta(y(x_{k-3}))\theta(x_{k-2})}$, $D_{k-3}(x_{k-4})^{\theta(y(x_{k-4}))\theta(x_{k-3})\theta(x_{k-2})}$, $\cdots$, $D_2(x_1)^{\theta(y(x_1))\theta(x_2)\cdots \theta(x_{k-2})}$ and $D_1^{\theta(x_1)\theta(x_2)\cdots \theta(x_{k-2})}$ with $\theta(x_{k-1})$, merge the same literals and get their instances $D_{k-2}(x_{k-3})^{\theta(y(x_{k-3}))\theta(x_{k-2})\theta(x_{k-1})}$, $D_{k-3}(x_{k-4})^{\theta(y(x_{k-4}))\theta(x_{k-3})\theta(x_{k-2})\theta(x_{k-1})}$, $\cdots$, \\$D_2(x_1)^{\theta(y(x_1))\theta(x_2)\cdots \theta(x_{k-2})\theta(x_{k-1})}$ and $D_1^{\theta(x_1)\theta(x_2)\cdots \theta(x_{k-2})\theta(x_{k-1})}$. In this case, $k$ extended clauses are gotten, i.e., $D_{k}(x_{k-1})^{\theta(y(x_{k-1}))}$, $D_{k-1}(x_{k-2})^{\theta(y(x_{k-2}))\theta(x_{k-1})}$, $D_{k-2}(x_{k-3})^{\theta(y(x_{k-3}))\theta(x_{k-2})\theta(x_{k-1})}$, $\cdots$, $D_2(x_1)^{\theta(y(x_1))\theta(x_2)\cdots \theta(x_{k-2})\theta(x_{k-1})}$ and $D_1^{\theta(x_1)\theta(x_2)\cdots \theta(x_{k-2})\theta(x_{k-1})}$.

\item[-]Therefore, the extended literals are $x_1^{\theta(x_1)\theta(x_2)\cdots \theta(x_{k-2})\theta(x_{k-1})}$, $x_2^{\theta(y(x_1))\theta(x_2)\cdots \theta(x_{k-2})\theta(x_{k-1})}$, $\cdots$, $x_{k-1}^{\theta(y(x_{k-2}))\theta(x_{k-1})}$, $\neg$$x_{k-1}^{\theta(y(x_{k-1}))}$. For simply representation, denote $\theta_i$ as the substitution of $x_i$ above, that is, $\theta_1=\theta(x_1)\theta(x_2)\cdots \theta(x_{k-2})\theta(x_{k-1})$, $\theta_2=\theta(y(x_1))\theta(x_2)\cdots \theta(x_{k-2})\theta(x_{k-1})$, $\cdots$, $\theta_{k-1}=\theta(y(x_{k-2}))\theta(x_{k-1})$, $\theta_k=\theta(y(x_{k-1}))$, and hence the extended literals are $x_1^{\theta_1}$, $x_2^{\theta_2}$, $\cdots$, $x_{k-1}^{\theta_{k-1}}$, $\neg$$x_{k-1}^{\theta_{k}}$. Connect the extended literals as the boundary line, which divided the extended clauses $D_i^{\theta_i}$($1 \leq i \leq k$) into two main parts, i.e., $D_i^{\theta_i+}$ and $D_i^{\theta_i-}$. Meanwhile, for each literal $x \in D_i^{\theta_{i}}(2 \leq i \leq k)$, if there exists a literal $x_j \in \{x_1^{\theta_1}$, $x_2^{\theta_2}$, $\cdots$, $x_{k-1}^{\theta_{k-1}}\}$, such that $x_j$ is complementary with $x$, then put $x$ in $D_i^{\theta_i-}$, and put $x$ on the same horizontal line as $x_j$ for intuitive representation. In this case, the separated part $D_i^{\theta_i-}$ ($i=1, 2, \cdots, k$) are taken together to form a contradiction $\bigwedge_{i=1}^kD_i^{\theta_i-}$ based on standard extension of $D_1, D_2,\cdots, D_k $, and the remaining parts $D_i^{\theta_i+}$ ($i=1, 2, \cdots, k $) extracted and merged the same literal to form a contradiction separation clause based on standard extension $\bigvee_{i=1}^kD_i^{\theta_i+}$ of $D_1, D_2, \cdots, D_k$. The schematic diagram of the contradictions separation based on standard extension is shown in Figure \ref{fig2}. 
\end{enumerate}

\begin{figure}[hbt]\centering
  \includegraphics[width=0.55\textwidth]{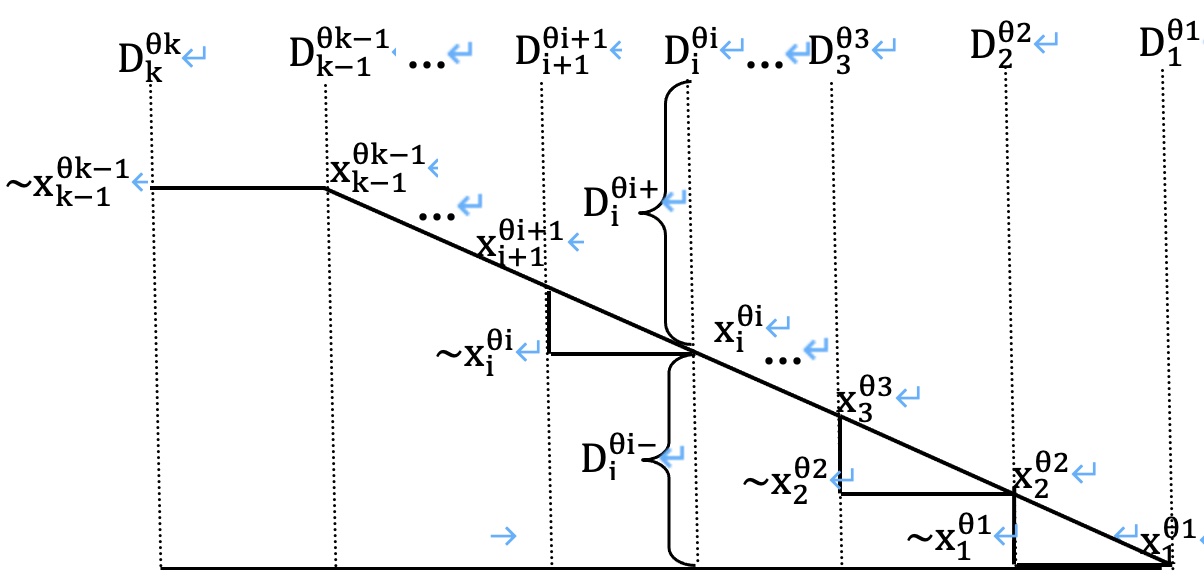} 
  \caption{Contradiction separation based on standard extension in first-order logic}
  \label{fig2}
\end{figure}

This process systematically constructs a sequence of extended clauses and corresponding substitutions, ensuring that all variable dependencies are resolved through unification.  
The resulting contradiction $\bigwedge_{i=1}^{k} D_i^{\theta_i^-}$ represents the joint unsatisfiable core generated by successive extensions, while the clause $\bigvee_{i=1}^{k} D_i^{\theta_i^+}$ denotes the corresponding contradiction separation clause under the Standard Extension.  
In essence, this procedure generalizes the propositional contradiction separation mechanism to first-order logic, where unification governs the interaction between literals and guarantees that the construction remains both sound and complete.  
The schematic illustration of this process is provided in Figure \ref{fig2}.

\begin{definition}\label{Def5.1}
(\textbf{Contradiction and the contradiction separation based on the standard extension}) In first-order logic, let $S = \{C_1, C_2,\cdots, C_n\}$ be a  clauses set, $D_1, D_2,\cdots, D_m (m \leq n)$ be the clauses that participate in the contradictions construction  based on standard extensions in $S$,  $\theta_i (1 \leq i \leq m)$ be a substitution of $D_i$. If it satisfies
\begin{enumerate}
\item[(1)] $D_i^{\theta_i} = D_i^{\theta_i +} \vee D_i^{\theta_i -}(1 \leq i \leq m)$,
\item[(2)]  $D_1^{\theta_1 -}= x_1^{\theta_1}$, $D_2^{\theta_2 -}=  \neg$$x_1^{\theta_1} \vee x_2^{\theta_2}$, $D_i^{\theta_i -}=(x_i^{\theta_i} \vee  \neg$$x_{i-1}^{\theta_{i-1}}) \vee \bigvee_{x \in D^{i}}x (3 \leq i \leq m -1)$, $D_m^{\theta_m -}= \neg$$x_{m-1}^{\theta_{m-1}} \vee \bigvee_{x\in D^m}x$, where $D^j \subseteq \bigcup_{h=1}^{j-2}\{ \neg$$x_h^{\theta_h}\}(3 \leq j \leq m) $, 
\end{enumerate}
then $\bigwedge_{i=1}^mD_i^{\theta_i-}$ is called a contradiction based on standard extension of $D_1^{\theta_1},\cdots, D_m^{\theta_m}$, $R_s (D_1^{\theta_1},\cdots, D_m^{\theta_m}) = \bigvee_{i=1}^mD_i^{\theta_i+}$ is called a contradiction separation  based on standard extension of $D_1^{\theta_1},\cdots, D_m^{\theta_m}$. The literal $x_i^{\theta_i}$ in $D_i^{\theta_i}$($1 \leq i \leq m -1$) is called the extended literal, and $D_i^{\theta_i}$($1 \leq i \leq m$) is called its extended clause accordingly. The literals in $\{x_1^{\theta_1}, x_2^{\theta_2},\cdots, x_{k-1}^{\theta_{k-1}}\}$ are called the literals on the main boundary, and the literals in  $\{ \neg$$x_1^{\theta_1},  \neg$$x_2^{\theta_2},\cdots,  \neg$$x_{k-1}^{\theta_{k-1}}\}$  are called the literals on the secondary boundary.  
\end{definition}

\begin{remark}\label{rmk5.1}
\begin{enumerate}
\item[(1)] In contradiction separation based on the Standard Extension for first-order logic, the same clause may participate in multiple extension steps, resulting in repeated use of extended clauses.

\item[(2)] When several clauses are involved in the Standard Extension–based contradiction separation and a standard contradiction is formed, that contradiction is removed, and a new clause is generated by disjoining the remaining parts to produce a contradiction-separation clause.  
In the special case where only two clauses are involved, the resulting contradiction separation coincides with the classical binary resolvent.  
Hence, contradiction separation based on the Standard Extension can be regarded as a natural generalization of binary resolution.
\end{enumerate}
\end{remark}

\begin{example}\label{Exa5.1} 
Let $S = \{C_1, C_2, C_3, C_4\}$ be a clauses set in first-order logic, $C_1$ = $\neg$$P(x_{11}, x_{12})$ $\vee$ $P(x_{12}, x_{11})$, $C_2$ = $\neg$$P(x_{21}, x_{22})$ $\vee   \neg$$P(x_{22}, x_{23}) \vee  P(x_{21}, x_{23})$, $C_3$ = $P(x_{31},f(x_{31}))$, $C_4$ = $\neg$$P (f(a_1), a_2) $, where $P$ is a predicate symbol, $x_{11}$, $x_{12}$, $x_{21}$, $x_{22}$, $x_{23}$, $x_{31}$, $x_{41}$ are variables symbols, $f$ is a function symbol, and $a_1, a_2$ are constant symbols. Then a contradiction of $S$ based on standard extension can be constructed as follows.

\begin{tabular}{lccc}\\\hline
  $C_4^{\sigma}$& $C_2^{\sigma}$& $C_1^{\sigma}$& $C_3^{\sigma}$\\ \hline
  & $ \neg$$P(a_1, a_2)$ & &  \\
 $ \neg$$P(f(a_1), a_2)$ & $P(f(a_1), a_2)$ & &  \\
  & $ \neg$$P(f(a_1), a_1)$ & $P(f(a_1), a_1)$ &  \\
 &   & $ \neg$$P(a_1, f(a_1))$& $P(a_1,f(a_1))$ \\\hline\\
\end{tabular}

Finally, a contradiction based on the standard extension of $C_4^{\sigma},C_2^{\sigma},C_1^{\sigma},C_3^{\sigma}$ is $P(a_1,f(a_1)) \wedge (P(f(a_1), a_1) \vee  \neg$$P(a_1,f(a_1))) \wedge ( \neg$$P(f(a_1), a_1) \vee P(f(a_1), a_2) \vee  \neg$$P(a_1, a_2))$ $\vee \neg$$P(f(a_1), a_2)$, and its corresponding contradiction separation is $R_s (C_3^{\sigma},C_1^{\sigma},C_2^{\sigma},C_4^{\sigma})$ = $ \neg$$P (a_1, a_2)$, where $\sigma = \{a_1 / x_{11}, f(a_1)/ x_{12}, f(a_1)/ x_{21}, a_1 / x_{22}, a_2 / x_{23}, a_1 / x_{31}\}$.
\end{example}

\begin{theorem}\label{thm5.1} 
The contradiction based on the standard extension in first-order logic is a standard contradiction.
\begin{proof}
Assuming that the vacant part of the $i^{th}$ row of the contradiction based on standard extension are added by $\neg$$y(x_{i-1})^{\theta(y(x_{i-1}))}$, $i=3,\cdots,k$, hence $D_i^{-}$ becomes to $D_i^{-*}$, that is, $D_i^{-*}= \{x_i^{\theta(y(x_{i-1}))}, \neg$$y(x_{i-1})^{\theta(y(x_{i-1}))},\cdots,  \neg$$y(x_{1})^{\theta(y_1)}\}$. In this case, $D_k^{-*}\times D_{k-1}^{-*}\times \cdots \times D_{1}^{-*}$ is a standard contradiction. $D_k^{-}\times D_{k-1}^{-}\times \cdots \times D_{1}^{-}$ is a cartesian subset of $D_k^{-*}\times D_{k-1}^{-*}\times \cdots \times D_{1}^{-*}$, so $D_k^{-}\times D_{k-1}^{-}\times \cdots \times D_{1}^{-}$  is a standard contradiction, that is, the contradiction based on the standard extension in first-order logic is a standard contradiction.
\end{proof} 
\end{theorem}

\begin{remark}\label{rmk5.2} 
\begin{enumerate}
\item[(1)] In the contradictions separation based on standard extension in first-order logic, if the clauses participating in the contradictions separation are $C_1$, $C_2$, $\cdots$, $C_k$, the contradiction separation clause is $R_s (C_1^{\theta_1}, C_2^{\theta_2},\cdots, C_k^{\theta_k})$, where $\theta_i (1 \leq i \leq k)$ is a substitution of $C_i$. Many factors can affect the result of $R_s (C_1^{\theta_1}, C_2^{\theta_2},\cdots, C_k^{\theta_k})$ such as the order of the extended clauses $C_1^{\theta_1}$, $C_2^{\theta_2}$, $\cdots$, $C_k^{\theta_k}$, extended literals in every extended clauses, the substitution $\theta_i$, and so on. If one of these factors is different, the result of contradiction separation may be different, and may not even constitute a contradiction separation based on standard extension.
\item[(2)] According to the contradiction separation construction method based on standard extension in first-order logic, the associative law does not hold. A counterexample is shown as follows. 
\end{enumerate}
\end{remark}

\begin{example}\label{Exa5.1.2} 
Let $S = \{C_1, C_2, C_3, C_4, C_5\}$ be a clauses set in first-order logic, $C_1$ = $\neg$$P(x_{11}, x_{12}) \vee P(x_{12}, x_{11}) $, $C_2$ = $\neg$$P(x_{21}, x_{22}) \vee  \neg$$P(x_{22}, x_{23}) \vee P(x_{21}, x_{23})$, $C_3$ = $P(x_{31},f(x_{31}))$, $C_4$ = $\neg$$P(f(a_1), a_2)$, $C_5$ = $P(a_1, a_2) $, where $P$ is a predicate symbol, $x_{11}, x_{12}, x_{21}, x_{22}, x_{23}, x_{31}, x_{41}$ is the variable symbols, $f$ is a function symbol, $a_1, a_2$ are constant symbols. According to Example \ref{Exa4.1},  there exists a contradiction separation of $C_3, C_1, C_2, C_4$, that is, $R_s(C_3^{\sigma_1}, C_1^{\sigma_1}, C_2^{\sigma_1}, C_4^{\sigma_1}) =  \neg$$P(a_1, a_2) $, where $\sigma= \{a_1 / x_{11}, f(a_1)/ x_{12}, f(a_1)/ x_{21}, a_1 / x_{22}, a_2 / x_{23}, a_1 / x_{31}\}$. Hence $R_s(R_s(C_3^{\sigma_1}, C_1^{\sigma_1}, C_2^{\sigma_1}$, $C_4^{\sigma_1}), C_5)$ = $\emptyset $. Since $C_4$ and $C_5$ are unit clauses, $R_s(C_1, C_2, C_4, C_5)$ does not exist. Therefore, $R_s(R_s(C_3^{\sigma_1}, C_1^{\sigma_1}, C_2^{\sigma_1}, C_4^{\sigma_1}), C_5) \neq R_s(C_3^{\sigma_1}, R_s(C_1^{\sigma_1}, C_2^{\sigma_1}, C_4^{\sigma_1}), C_5)$.  
\end{example}

\begin{definition}\label{Def5.2}
\textbf{(Refutation Sequence under the Standard Extension in First-Order Logic)}  
Let $S = \{C_1, C_2, \ldots, C_n\}$ be a set of clauses in first-order logic, and let  
$w = \{\phi_1, \phi_2, \ldots, \phi_t\}$ denote a deduction sequence of contradiction separation based on the Standard Extension (SE) from $S$ to $\phi_t$.  
Each formula $\phi_i$ $(i = 1, 2, \ldots, t)$ satisfies one of the following conditions:
\begin{enumerate}
    \item[(1)] $\phi_i \in S$, or
    \item[(2)] there exist indices $r_1, r_2, \ldots, r_{k_i} < i$ and a substitution $\theta_i$ such that  
    \[
    \phi_i = R_s(\phi_{r_1}^{\theta_i}, \phi_{r_2}^{\theta_i}, \ldots, \phi_{r_{k_i}}^{\theta_i}),
    \]
    where $R_s$ denotes the inference rule of contradiction separation under the Standard Extension.
\end{enumerate}
If $\phi_t = \emptyset$, then $w$ is called a \textit{refutation sequence of contradiction separation based on the Standard Extension} for $S$.
\end{definition}

\begin{example}\label{Exa5.2} 

Let $S = \{C_1, C_2, C_3, C_4, C_5, C_6, C_7\}$ be a clauses set in first-order logic, $C_1$ = $\neg$$P_1 (x_1) \vee P_2 (x_1) \vee P_3 (x_1,f_1(x_1))$, $C_2$ = $\neg$$P_1 (x_2) \vee P_2 (x_2) \vee P_4 (f_1 (x_2))$, $C_3$ = $P_5 (a)$, $C_4$ = $P_1 (a)$, $C_5$ = $\neg$$P_3 (a, x_3) \vee P_5 (x_3)$, $C_6$ = $\neg$$P_5 (x_4) \vee  \neg$$P_2 (x_4) $, $C_7$ = $\neg$$P_5 (x_5) \vee$  $\neg$$P_4 (x_5)$, where $P_1, P_2, \cdots, P_5$ are predicate symbols, $x_1, x_2,\cdots, x_5$ is variable symbols, $a$ is a constant symbol, and $f_1$ is a function symbol. Using the method of contradictions separation based on standard extension, the following deduction sequence can be obtained.

The first contradiction separation based on standard extension can be constructed as follows.

\begin{tabular}{lccccc}\\\hline
 $C_6 ^{\sigma_1}$ &$C_2^{\sigma_1}$ &$C_7 ^{\sigma_1}$ & $C_5 ^{\sigma_1}$&$C_1 ^{\sigma_1}$ &$C_4 ^{\sigma_1}$\\ \hline
  $ \neg$$P_2 (a)$& $P_2 (a)$& & & &\\
  & $P_4 (f_1 (a))$& $ \neg$$P_4 (f_1 (a))$& &&\\
  $ \neg$$P_5 (a)$& & $ \neg$$P_5 (f_1 (a))$& $P_5 (f_1 (a))$&$P_2 (a)$&\\
  & & & $ \neg$$P_3 (a, f_1 (a))$&$P_3 (a,f_1 (a))$&\\
  & $ \neg$$P_1 (a)$& & &$ \neg$$P_1 (a)$&$P_1 (a)$ \\\hline\\
\end{tabular}

Therefore, the contradiction separation clause can be obtained, that is, $C_8$ = $R_s (C_6^{\sigma_1}, C_2^{\sigma_1}, C_7 ^{\sigma_1}, C_5^{\sigma_1}, C_1^{\sigma_1}, C_4^{\sigma_1})$ = $P_2 (a) $, where $\sigma_1 $ = $\{a / x_1$, $a / x_2$, $f_1 (a)/ x_3$, $a / x_4$, $f_1 (a)/ x_5\}$.

The second contradiction separation based on standard extension can be constructed as follows.

\begin{tabular}{lcc}\\\hline
  $C_3^{\sigma_2}$& $C_6^{\sigma_2}$& $C_8^{\sigma_2}$\\
  \hline
  $P_5 (a)$ & $ \neg$$P_5 (a)$ & \\
  & $ \neg$$P_2 (a)$ & $P_2 (a)$\\\hline\\
\end{tabular}

Therefore, the contradiction separation clause can be obtained, that is,  $C_9 = R_s (C_3^{\sigma_2}, C_6^{\sigma_2}, C_8^{\sigma_2}) =\emptyset$, where $\sigma_2=\{a / x_4\}$. Therefore, a refutation of contradiction separation based on standard extension of $S$ is obtained, i.e., $C_1$, $C_2$, $C_3$, $C_4$, $C_5$, $C_6$, $C_7$, $C_8=P_2 (a)$, $C_9=\emptyset$.
\end{example}

\begin{theorem}\label{thm5.2} 
Suppose $S=\{C_1, C_2,\cdots, C_k\}$  is a clauses set in first-order logic, $R_s (C_1^{\sigma_1}, C_2^{\sigma_2},\cdots, C_k^{\sigma_k})$ is a contradiction separation based on standard extension of $S$, then there exists a linear resolution deduction of $S$ with $C_k$ as the top clause and $C_{k- 1},\cdots, C_1$ as the side clauses, such that the linear resolvent is $R_{k,k-1, k-2,\cdots,2, 1} (\cdots$ $R_{k,k-1, k-2} (R_{k,k-1}(C_k ^{\sigma_k}, C_{k- 1}^{\sigma_{k- 1}}), C_{k- 2}^{\sigma_{k- 2}}) \cdots, C_1^{\sigma_1})$ = $R_s (C_1 ^{\sigma_1}, C_2^{\sigma_2},\cdots, C_k^{\sigma_k})$, where $\sigma_i (1 \leq i \leq k)$ is a substitution of $C_i$.
\begin{proof}
The theorem follows similar to Theorem \ref{thm4.6} (propositional case), adapting unification and
 substitutions throughout.	
\end{proof} 
\end{theorem}

\begin{remark}\label{rmk5.4}
The inverse proposition of Theorem \ref{thm5.2} does not necessarily hold, that is, if $R_{k,k-1, k-2,\cdots,2, 1} (\cdots$ $R_{k,k-1, k-2} (R_{k,k-1}(C_k ^{\sigma_k}, C_{k- 1}^{\sigma_{k- 1}}), C_{k- 2}^{\sigma_{k- 2}}) \cdots, C_1^{\sigma_1})$ is a linear resolvent of the clauses set $S = \{C_1, C_2,\cdots, C_k\} $, but there does not necessarily exist a standard  contradiction separation based on standard extension of $C_1^{\sigma_1}, C_2^{\sigma_2},\cdots, C_k^{\sigma_k}$, where $\sigma_i (1 \leq i \leq k)$ is a substitution of $C_i$. A counterexample is shown as follows.
\end{remark}

\begin{example}\label{Exa4.3}
Let $S = C_1 \wedge C_2  \wedge C_3 \wedge C_4$ be a clauses set in first-order logic, $C_1$ = $P_1 (x_1) \vee P_2 (x_2) \vee P_3 (x_3) \vee P_4 (x_4)$, $C_2$ = $\neg$$P_1 (a_1)$, $C_3$ = $\neg$$P_2 (a_2)$, $C_4$ = $\neg$$P_3 (a_3) $, where $P_1, P_2, P_3, P_4$ are predicate symbols, $x_1, x_2, x_3, x_4$ are variable symbols, $a_1, a_2, a_3$ are constant symbols. Then there exist a linear resolution deduction $R_{1,2,3,4}(R_{1,2,3} (R_{1,2}$ $(C_1^{\sigma}, C_2^{\sigma})$, $C_3^{\sigma}), C_4^{\sigma})$ = $P_4 (x_4) $, where $\sigma =\{a_1 / x_1, a_2 / x_2, a_3 / x_3\} $.

However, there does not exist a deduction of contradiction separation based on standard extension of $C_4, C_3, C_2, C_1$, such that $R_s (C_4, C_3, C_2, C_1)$ = $P_4(x_4)$.
\end{example}

\subsection{\label{sec:level5.2}Soundness and Completeness} 
In first-order logic, since the deduction of contradiction separation based on standard extension is a special case of the deduction of contradiction separation, and contradiction separation method is sound, hence the soundness of the deduction of contradiction separation based on standard extension holds accordingly, that is, Theorem \ref{thm5.3} holds.

\begin{theorem}\label{thm5.3}
(\textbf{Soundness}) Suppose $S =\{C_1,\cdots, C_n\}$ is a clauses set in first-order logic, $\phi_ 1, \phi_ 2,\cdots, \phi_ t$ is a deduction of contradiction separation based on standard extension from $S$ to $\phi_ t$. If $\phi_t = \emptyset$, then $S$ is unsatisfiable.
\end{theorem}

\begin{theorem}\label{thm5.4} 
(\textbf{Completeness}) Suppose $S =\{C_1,\cdots, C_n\}$ is a clauses set in first-order logic. If $S$ is unsatisfiable, then there exists a refutation of contradiction separation based on standard extension of $S$.
\begin{proof}
In first-order logic, on the one hand, the binary resolution deduction is a special case of the deduction of contradiction separation based on standard extension. Since the binary resolution deduction is complete, then the completeness of the contradiction separation based on standard extension holds accordingly. On the other hand, there exist more than two clauses in the dynamic deduction process. Therefore, the completeness of deduction of contradiction separation based on standard extension holds.
\end{proof}
\end{theorem}

\begin{example}\label{Exa5.4}
Let $S = \{C_1, C_2, C_3, C_4, C_5\}$ be a clauses set in first-order logic, $C_1 = P(a_1)$, $C_2$ = $\neg$$D(x_{21}) \vee L (a_1, x_{22})$, $C_3$ = $\neg$$P(x_{31})$ $\vee \neg$$Q(x_{32})$ $\vee \neg$$L(x_{31}, x_{32})$, $C_4$ = $D(a_2)$, $C_5$ = $Q(a_2)$, where $P, D, L, Q$ are predicate symbols, $x_{21}, x_{22}, x_{31}, x_{32}$ are variable symbols, and $a_1, a_2$ are constant symbols. It is easy to verify $S$ is unsatisfiable, and a refutation $w$ of contradiction separation based on standard extension of $S$ can be obtained as follows.

\begin{enumerate}
\item[(1)] $C_1$ = $P(a_1)$
\item[(2)] $C_2$ = $\neg$$D(x_{21}) \vee  L(a_1, x_{22})$

\item[(3)] $C_3$ = $\neg$$P(x_{31}) \vee \neg$$Q(x_{32}) \vee \neg$$L(x_{31}, x_{32})$

\item[(4)] $C_4$ = $D(a_2)$

\item[(5)] $C_5$ = $Q(a_2)$
 
 ------------
\item[(6)] $C_6$ = $R_s (C_1^{\theta_1}, C_3^{\theta_1}, C_5 ^{\theta_1})$ =  $\neg$$L(a_1, a_2) $, where $\theta_1=\{a_1 / x_{31}, a_2 / x_{32}\}$

\item[(7)] $C_7 = R_s (C_6^{\theta_2}, C_2^{\theta_2}, C_4 ^{\theta_2})= \emptyset$, where $\theta_2=\{a_2 / x_{21}, a_2 / x_{22}\}$
\end{enumerate}

Therefore, $w=\{C_1, C_2, C_3, C_4, C_5, C_6, C_7\}$ is a refutation of contradiction separation based on standard extension of $S$.
\end{example}

The above example illustrates the correctness and practical realizability of the deduction process under the Standard Extension in first-order logic.  
Together, Theorems~\ref{thm5.3} and~\ref{thm5.4} establish that the contradiction separation method based on the Standard Extension is both sound and complete, ensuring that every unsatisfiable clause set can be refuted, while no satisfiable set produces a false contradiction.  
These properties provide the theoretical foundation for the unified algorithm introduced in the following subsection, enabling efficient automated reasoning through systematic construction and separation of contradictions.

\subsection{\label{sec:level5.3}A Unified Dynamic Deduction Method Based on the Standard Extension in First-Order Logic}

This section introduces a unified dynamic deduction framework for contradiction separation under the Standard Extension in first-order logic.  
The method provides a systematic procedure for determining whether a formula is satisfiable or unsatisfiable.  
In particular, when a formula is satisfiable, the approach can also construct a corresponding satisfying example.

The unified deduction procedure based on the Standard Extension in first-order logic is outlined as follows.

\begin{enumerate}
\item[Step 0.] (Preprocessing)
\begin{enumerate}
\item[(1)] Let $S=C_1 \wedge C_2 \wedge \cdots \wedge C_n$ be a clauses set in first-order logic. Rename variable symbols of the clauses if there exist common variable symbols between all the clauses in $S$, and get a new clauses set denoted as $S_0$.
\item[(2)] Delete the subsumed clauses and tautologies in $S_0$, and get a new clauses set denoted as $S_1$.
\end{enumerate}
\item[Step 1.] Select a clause in $S_1$ denoted as $D_1$, and select a literal in $D_1$ denoted as $x_1$. According to $x_1$, select a clause $D_2(x_1)$ in $S_1$, and $D_2(x_1)$ satisfies: there exist the literal $y(x_1) \in D_2(x_1)$, and substitutions $\theta(x_1)$ and $\theta(y(x_1))$ such that $x_1^{\theta(x_1)}$ and $y(x_1)^{\theta(y(x_1))}$ is a complementary pair, i.e., $x_1^{\theta(x_1)}$ = $\neg$$y(x_1)^{\theta(y(x_1))}$. 
\begin{enumerate}
\item[(1)] Substitute $D_1$ with $\theta(x_1)$, and merge the same literals to obtain a new clause $D_1^{\theta(x_1)}$.
\item[(2)]  Substitute $D_2(x_1)$ with $\theta(y(x_1))$, and merge the same literals to obtain a new clause $D_2(x_1)^{\theta(y(x_1))}$.
\item[(3)] Divide the clause $D_1^{\theta(x_1)}$ into two parts: $D_1^{\theta(x_1)-}$ and $D_1^{\theta(x_1)+}$, where $D_1^{\theta(x_1)-}$ = $x_1^{\theta(x_1)}$, $D_1^{\theta(x_1)+}$ = $D_1^{\theta(x_1)}-x_1^{\theta(x_1)-}$, which means that $D_1^{\theta(x_1)+}$ is gotten by deleting $D_1^{\theta(x_1)-}$ from $D_1^{\theta(x_1)}$.
\end{enumerate}

In this case, the extended literal is $x_1^{\theta(x_1)}$, and the extended clauses are $D_1^{\theta(x_1)}$, $D_2(x_1)^{\theta(y(x_1))}$.
\item[Step 2.] Select a literal in $D_2(x_1)^{\theta(y(x_1))}-y(x_1)^{\theta(y(x_1))}$ denoted as $x_2$. According to $x_2$, select a clause $D_3(x_2)$ in $S_1$, and $D_3(x_2)$ satisfies: there exist the literal $y(x_2) \in D_3(x_2)$, and substitutions $\theta(x_2)$ and $\theta(y(x_2))$ such that $x_2^{\theta(x_2)}$ and $y(x_2)^{\theta(y(x_2))}$ is a complementary pair, i.e., $x_2^{\theta(x_2)}$ = $\neg$$y(x_2)^{\theta(y(x_2))}$. 
\begin{enumerate}
\item[(1)]  Substitute $D_1^{\theta(x_1)}, D_2(x_1)^{\theta(y(x_1))}$ with $\theta(x_2)$, and merge the same literals to obtain the new clauses $D_1^{\theta(x_1)\theta(x_2)}$, $D_2(x_1)^{\theta(y(x_1))\theta(x_2)}$.
\item[(2)] Substitute $D_3 (x_2)$ with $\theta(y(x_2))$, and merge the same literals to obtain a new clause $D_3 (x_2)^{\theta(y(x_2))}$.
\item[(3)] Divide the clause $D_2(x_1)^{\theta(y(x_1))\theta(x_2)}$ into two parts: $D_2(x_1)^{\theta(y(x_1))\theta(x_2)-}$ and $D_2(x_1)^{\theta(y(x_1))\theta(x_2)+}$, where $D_2(x_1)^{\theta(y(x_1))\theta(x_2)-}= x_2^{\theta(x_2)} \vee y(x_1)^{\theta(y(x_1))\theta(x_2)}$. 

$D_2(x_1)^{\theta(y(x_1))\theta(x_2)+}=D_2(x_1)^{\theta(y(x_1))\theta(x_2)}-D_2(x_1)^{\theta(y(x_1))\theta(x_2)-}$, which means that $D_2(x_1)^{\theta(y(x_1))\theta(x_2)+}$ is gotten by deleting $D_2(x_1)^{\theta(y(x_1))\theta(x_2)-}$ from $D_2(x_1)^{\theta(y(x_1))\theta(x_2)}$.
\end{enumerate}

In this case, the extended literals are $x_1^{\theta(x_1)\theta(x_2)}$ and $x_2^{\theta(x_2)}$, and the extended clauses are $D_1^{\theta(x_1)\theta(x_2)}$, $D_2(x_1)^{\theta(y(x_1))\theta(x_2)}$ and $D_3 (x_2)^{\theta(y(x_2))}$.

\item[$\cdots$]  $ $

\item[Step $i(i \geq 3) $] Select a literal in $D_i(x_{i-1})^{\theta(y(x_{i-1}))}-\bigvee_{j=1}^{i-2}y(x_j)^{\theta(y(x_j))\theta(x_{j+1})\cdots\theta(x_{i-1})}-y(x_{i-1})^{\theta(y(x_{i-1}))}$ denoted as $x_i$. According to $x_i$, select a clause $D_{i+1}(x_i)$ in $S_1$, and $D_{i+1}(x_i)$ satisfies: there exist the literal $y(x_i) \in  D_{i + 1 }(x_i) $, and substitutions $\theta (x_i)$ and $\theta (y(x_i)) $ such that $x_i^{\theta(x_i)}$  and $y(x_i)^{\theta(y(x_i))}$ is a complementary pair, i.e., $x_i^{\theta(x_i)}$ = $\neg$$y(x_i)^{\theta(y(x_i))}$. 
\begin{enumerate}
\item[(1)]  Substitute $D_1^{\theta(x_1)\theta(x_2)\cdots\theta(x_{i-1})}$, $D_2(x_1)^{\theta(y(x_1))\theta(x_2)\cdots\theta(x_{i-1})}$, $\dots$, $D_i(x_{i-1})^{\theta(y(x_{i-1}))}$ with $\theta(x_{i})$, and merge the same literals to obtain the new clauses $D_1^{\theta(x_1)\theta(x_2)\cdots\theta(x_{i-1})\theta(x_{i})}$, $D_2(x_1)^{\theta(y(x_1))\theta(x_2)\cdots\theta(x_{i-1})\theta(x_{i})}$, $\dots$, $D_i(x_{i-1})^{\theta(y(x_{i-1}))\theta(x_{i})}$.
\item[(2)] Substitute $D_{i+1}(x_{i})$ with $\theta(y(x_i)) $, and merge the same literals to obtain a new clause $D_{i+1}(x_{i})^{\theta(y(x_i))}$.
\item[(3)] Divide the clause $D_{i}(x_{i-1})^{\theta(y(x_{i-1}))\theta(x_i)}$ into two parts: $D_{i}(x_{i-1})^{\theta(y(x_{i-1}))\theta(x_i)-}$ and $D_{i}(x_{i-1})^{\theta(y(x_{i-1}))\theta(x_i)+}$, where $D_{i}(x_{i-1})^{\theta(y(x_{i-1}))\theta(x_i)-}$ = $(x_i^{\theta(x_i)} \vee y(x_{i-1})^{\theta(y(x_{i-1}))\theta(x_i)})$ $\vee D_i^0$, $D_i^0$ = $\bigvee_{x\in D^i}x$, and $D^i \subseteq \{y(x_j)^{\theta(y(x_j))\theta(x_{j+1})\cdots\theta(x_{i})}|j=1,2, \cdots, i-2\}$. 

$D_{i}(x_{i-1})^{\theta(y(x_{i-1}))\theta(x_i)+}$ = $D_{i}(x_{i-1})^{\theta(y(x_{i-1}))\theta(x_i)}- D_{i}(x_{i-1})^{\theta(y(x_{i-1}))\theta(x_i)-}$, which means that $D_{i}(x_{i-1})^{\theta(y(x_{i-1}))\theta(x_i)+}$ is gotten by deleting $D_{i}(x_{i-1})^{\theta(y(x_{i-1}))\theta(x_i)-}$ from $D_{i}(x_{i-1})^{\theta(y(x_{i-1}))\theta(x_i)}$.
\end{enumerate}

In this case, the extended literals are $x_1^{\theta(x_1)\cdots\theta(x_i)}$, $x_2^{\theta(x_2)\cdots\theta(x_i)}$, $\cdots$, $x_{i-1}^{\theta(x_{i-1})\theta(x_i)}$, and $x_i^{\theta(x_i)}$, and the extended clauses are $D_1^{\theta(x_1)\theta(x_2)\cdots\theta(x_i)}$, $D_2(x_1)^{\theta(y(x_1))\theta(x_2)\cdots\theta(x_i)}$, $\cdots$, $D_{i}(x_{i-1})^{\theta(y(x_{i-1}))\theta(x_i)}$, and $D_{i+1}(x_i)^{\theta(y(x_i))}$.

\item[$\cdots$ ] $ $ 

\item[Step $k-1$.] Select a literal in $D_{k-1}(x_{k-2})^{\theta(y(x_{k-2}))}-\bigvee_{j=1}^{k-3}y(x_j)^{\theta(y(x_j))\theta(x_{j+1})\cdots\theta(x_{k-2})}-y(x_{k-2})^{\theta(y(x_{k-2}))}$  denoted as $x_{k-1}$. According to $x_{k-1}$, select a clause $D_k (x_{k-1})$ in $S_1$, and $D_k (x_{k- 1})$ satisfies: there exist the literal $y(x_{k-1}) \in D_k (x_{k-1}) $, and substitutions $\theta(x_{k-1})$ and $\theta(y(x_{k-1}))$ such that $x_{k-1}^{\theta(x_{k-1})}$  and $y(x_{k-1})^{\theta(y(x_{k-1}))}$ is a complementary pair, i.e., $x_{k-1}^{\theta(x_{k-1})}$ = $\neg$$y(x_{k-1})^{\theta(y(x_{k-1}))}$.
\begin{enumerate}
\item[(1)] Substitute $D_1^{\theta(x_1)\theta(x_2)\cdots\theta(x_{k-2})}$, $D_2(x_1)^{\theta(y(x_1))\theta(x_2)\cdots\theta(x_{k-2})}$, $\dots$, $D_{k-1}(x_{k-2})^{\theta(y(x_{k-2}))}$ with $\theta(x_{k-1})$, and merge the same literals to obtain the new clauses $D_1^{\theta(x_1)\theta(x_2)\cdots\theta(x_{k-2})\theta(x_{k-1})}$, $D_2(x_1)^{\theta(y(x_1))\theta(x_2)\cdots\theta(x_{k-2})\theta(x_{k-1})}$, $\dots$, $D_{k-1}(x_{k-2})^{\theta(y(x_{k-2}))\theta(x_{k-1})}$.
\item[(2)] Substitute $D_k (x_{k-1})$ with $\theta (y(x_{k-1}))$, and merge the same literals to obtain a new clause $D_{k}(x_{k-1})^{\theta(y(x_{k-1}))}$.
\item[(3)] Divide the clause $D_{k-1}(x_{k-2})^{\theta(y(x_{k-2}))\theta(x_{k-1})}$  into two parts: $D_{k-1}(x_{k-2})^{\theta(y(x_{k-2}))\theta(x_{k-1})-}$ and $D_{k-1}(x_{k-2})^{\theta(y(x_{k-2}))\theta(x_{k-1})+}$, where $D_{k-1}(x_{k-2})^{\theta(y(x_{k-2}))\theta(x_{k-1})-}= (x_{k-1}^{\theta(x_{k-1})} \vee y(x_{k-2})^{\theta(y(x_{k-2}))\theta(x_{k-1})})\vee D_{k-1}^0$, $D_{k-1}^0=\bigvee_{x\in D^{k-1}}x$, and $D^{k-1}\subseteq \{y(x_j)^{\theta(y(x_j))\theta(x_{j+1})\cdots\theta(x_{k-1})}$ $|j=1,2, \cdots, k-3\}$. 

$D_{k-1}(x_{k-2})^{\theta(y(x_{k-2}))\theta(x_{k-1})+}= D_{k-1}(x_{k-2})^{\theta(y(x_{k-2}))\theta(x_{k-1})}- D_{k-1}(x_{k-2})^{\theta(y(x_{k-2}))\theta(x_{k-1})-}$, which means that $D_{k-1}(x_{k-2})^{\theta(y(x_{k-2}))\theta(x_{k-1})+}$ is gotten by deleting $D_{k-1}(x_{k-2})^{\theta(y(x_{k-2}))\theta(x_{k-1})-}$ from $D_{k-1}(x_{k-2})^{\theta(y(x_{k-2}))\theta(x_{k-1})}$.
\end{enumerate}

In this case, the extended literals are $x_1^{\theta(x_1)\cdots\theta(x_{k-1})}$, $x_2^{\theta(x_2)\cdots\theta(x_{k-1})}$, $\cdots$, $x_{k-2}^{\theta(x_{k-2})\theta(x_{k-1})}$, and $x_{k-1}^{\theta(x_{k-1})}$, and the extended clauses are $D_1^{\theta(x_1)\theta(x_2)\cdots\theta(x_{k-1})}$, $D_2(x_1)^{\theta(y(x_1))\theta(x_2)\cdots\theta(x_{k-1})}$, $\cdots$, $D_{k-1}(x_{k-2})^{\theta(y(x_{k-2}))\theta(x_{k-1})}$, and $D_{k}(x_{k-1})^{\theta(y(x_{k-1}))}$.

\item[Step $k$.] If it cannot be extended or no longer needs to be extended, then stop the extension. In this case, the clause $D_{k}(x_{k-1})^{\theta(y(x_{k-1}))}$ is divided into two parts: $D_{k}(x_{k-1})^{\theta(y(x_{k-1}))-}$ and $D_{k}(x_{k-1})^{\theta(y(x_{k-1}))+}$, where $D_{k}(x_{k-1})^{\theta(y(x_{k-1}))-}$ = $(y(x_{k-1})^{\theta(x_{k-1})}$ $\vee D_{k}^0$, $D_{k}^0$ = $\bigvee_{x\in D^{k}}x$, and $D^{k}\subseteq \{y(x_j)^{\theta(y(x_j))\theta(x_{j+1})\cdots\theta(x_{k-1})}|j=1,2, \cdots, k-2\}$.

$D_{k}(x_{k-1})^{\theta(y(x_{k-1}))+}$ = $D_{k}(x_{k-1})^{\theta(y(x_{k-1}))}- D_{k}(x_{k-1})^{\theta(y(x_{k-1}))-}$, which means that\\ $D_{k}(x_{k-1})^{\theta(y(x_{k-1}))+}$ is gotten by deleting $D_{k}(x_{k-1})^{\theta(y(x_{k-1}))-}$ from $D_{k}(x_{k-1})^{\theta(y(x_{k-1}))}$.
\end{enumerate}

For simply representation, the substitution of every extended clause above is denoted as $\theta_i (1 \leq i \leq k)$, that is, $\theta_1=\theta(x_1)\theta(x_2)\cdots \theta(x_{k-2})\theta(x_{k-1})$, $\theta_2=\theta(y(x_1))\theta(x_2)\cdots \theta(x_{k-2})\theta(x_{k-1})$, $\cdots$, $\theta_{k-1}=\theta(y(x_{k-2}))\theta(x_{k-1})$, $\theta_k=\theta(y(x_{k-1}))$. Finally, the contradiction $\bigwedge_{i=1}^kD_i(x_{i-1})^{\theta_i-}$ and the corresponding contradiction separation clause $R_s(D_1^{\theta_1}, \cdots,D_k^{\theta_k}) = \bigvee_{i=1}^kD_i(x_{i-1})^{\theta_i+}$ based on standard extension are obtained respectively.

Two cases exist as follows.
\begin{enumerate}
\item[(1)] $R_s(D_1^{\theta_1}, \cdots,D_k^{\theta_k}) = \emptyset$, then stop, and $S$ is unsatisfiable.

\item[(2)] $R_s(D_1^{\theta_1}, \cdots,D_k^{\theta_k}) \neq \emptyset $, two sub-cases exist as follows.
\begin{enumerate}
\item[1)] If $k \geq n$, all the clauses in $S_1$ are involved for contradiction constructing, and there exists a literal $y \in D_{j_0}(x_{j_0-1})^{\theta_{j_0}+} (1 \leq j_0 \leq n)$, and satisfies one of the following conditions.
\begin{enumerate}
\item [-] The predicate symbol of $y$ is not the same with all the predicate symbols in $\{y (x_{k-1}), \cdots, y (x_{j_0}), x_{j_0-1}, \cdots, x_1\}$,  or

\item [-] The predicate symbol of $y$ is complementary with some the predicate symbols in $\{y (x_{k-1}), \cdots, y (x_{j_0}), x_{j_0-1}, \cdots, x_1\}$, but all the complementary literals are ground literals with the different constant terms. 
\end{enumerate}

Then stop, and $S_1$ is satisfiable, $(y(x_{k-1}),\cdots, y(x_{j_0}), y, x _{j_0-1},\cdots, x_1)$ is a satisfiable example of $S_1$. Among them, if there exist the same clauses in the extended clause, the same literal is selected in the satisfiable example.

\item[2)] Otherwise, the contradiction separation $R_s(D_1^{\theta_1}, \cdots,D_k^{\theta_k})$ is gotten, and $S = S \cup \{R_s(D_1^{\theta_1}, \cdots,D_k^{\theta_k})\} $, then go to Step 0.
\end{enumerate}
\end{enumerate}

\begin{remark}\label{rmk5.5} 
\begin{enumerate}
\item[(1)] In first-order logic, when using the contradiction separation deduction method based on standard extension to verify a logical formula, it is necessary to give the priority of the literals and clauses for further improving its efficiency such as the number of literals in a clause, the proportion of logical symbols in the literal, the number of ground items, weights, etc. According to the priority, we can sort the literals and clauses, and select the first literal as the extended literal, which is used to be extended for contradiction and contradiction separation construction. Furthermore, the priority just gives the order of the selection of extended literals and clauses, it does not affect the completeness of the method.
\item[(2)] In the process of contradiction construction based on standard extension, $x_i$ and $y (x_i)$  are unified respectively by substitutions $\theta(x_i)$ and $\theta (y(x_i))$ such that  $x_i^{\theta (y(x_{i-1}))\theta(x_i)}$ = $\neg$$y(x_i)^{\theta (y(x_i))}$. In this case, $\theta(x_i)$ and $\theta(y (x_i))$ not only substitute their clauses $D_i (x_{i-1})$ and $D_{i+1}(x_i)$, but also should substitute the clauses $D_{i-1} (x_{i-2}),\cdots, D_{2}(x_{1}), D_1$. This is because the substitution $\theta(x_i)$ affects the literal $y(x_{i-1})^{\theta (y(x_{i-1}))}$ in $D_i (x_{i-1})^{\theta (y(x_{i-1}))}$, i.e., $y(x_{i-1})^{\theta (y(x_{i-1}))\theta(x_i)}$. But the literal $y(x_{i-1})^{\theta (y(x_{i-1}))}$ should be the complementary literal with $x_{i-1}^{\theta(x_{i-1})}$ in $D_{i-1} (x_{i-2})^{\theta(x_{i-1})}$, hence $\theta(x_i)$ should substitute $x_{i-1}^{\theta(x_{i-1})}$, such that $x_{i-1}^{\theta(x_{i-1})\theta(x_i)}$ = $\neg$$y(x_{i-1})^{\theta (y(x_{i-1}))\theta(x_i)}$. This is the similar case for $D_{i-2} (x_{i-3}),D_{i-1}(x_{i-2}), \cdots, D_1$. Such substitutions are called reverse substitutions in the contradiction constructing.
\end{enumerate}
\end{remark}

The procedure described above formalizes a complete operational framework for contradiction separation based on the Standard Extension in first-order logic.  
Through iterative clause extension, unification, and reverse substitution, the method dynamically constructs contradiction structures and their corresponding separation clauses, ensuring that every deduction step remains logically consistent.  
This unified approach not only guarantees soundness and completeness but also provides a practical foundation for efficient implementation in automated reasoning systems.  

According to this method, a unified algorithm for deduction of contradictions separation based on standard extension in first-order logic is shown in Algorithm \ref{agm2} as shown in the Appendix.

\begin{algorithm}[h]
   \caption{A unified algorithm for deduction of contradictions separation based on standard extension in first-order logic} \label{agm2}
   \BlankLine
   \KwIn{$S= \{C_1, C_2, \cdots, C_n\}$ is the clauses set in first-order logic.}
   \KwOut{Judge the unsatisfiability and satisfiability of $S$. If $S$ is satisfiable, and satisfies some conditions, then give its  satisfiable example.}
   Initiate: $T_0$ is the set of the conditions for stopping the contradiction extension; $T_1$ is the set of the conditions for judging satisfiability of $S$.\\
   \While {$P_s \neq 0$ $\&$ $P_s \neq 1$ }{
   $S_0 \leftarrow S$ by deleting the subsumed clauses and tautologies of $S$;\\
   $S_1 \leftarrow S_0$ by renaming teh variable symbols of all the clauses in $S_0$; k=1;\\
   Select $D_1$ is $S_1$, and $x_1$ in $D_1$;\\
   Extend $x_1$ and select $y(x_1)$ in $S_1$, denote its clause as $D_2(x_1)$, such that $x_1^{\theta(x_1)}$ = $\neg$$y(x_1)^{\theta(y(x_1))}$;\\
   The extended literal is $x_1^{\theta(x_1)}$, and the extended clauses are $D_1^{\theta(x_1)}$, $D_2(x_1)^{\theta(y(x_1))}$;\\
   $D_1^{\theta(x_1)}=D_1^{\theta(x_1)-}+D_1^{\theta(x_1)+}$, where $D_1^{\theta(x_1)-}$ = $x_1^{\theta(x_1)}$.\\
     \While {$T_0$ is not met}{
     $k=k+1$;\\
     Select $x_{k-1}$ in $D_{k-1}(x_{k-2})^{\theta(y(x_{k-2}))}-\bigvee_{j=1}^{k-3}y(x_j)^{\theta(y(x_j))\theta(x_{j+1})\cdots\theta(x_{k-2})}-y(x_{k-2})^{\theta(y(x_{k-2}))}$;\\
     Extend $x_{k-1}$ and select $y(x_{k-1})$ in $S_1$, denote its clause as $D_k(x_{k-1})$, such that $x_{k-1}^{\theta(x_{k-1})}$ = $\neg$$y(x_{k-1})^{\theta(y(x_{k-1}))}$;\\
     The extended literals are $x_1^{\theta(x_1)\cdots\theta(x_{k-1})}$, $x_2^{\theta(x_2)\cdots\theta(x_{k-1})}$, $\cdots$, $x_{k-2}^{\theta(x_{k-2})\theta(x_{k-1})}$, and $x_{k-1}^{\theta(x_{k-1})}$, and the extended clauses are $D_1^{\theta(x_1)\theta(x_2)\cdots\theta(x_{k-1})}$, $D_2(x_1)^{\theta(y(x_1))\theta(x_2)\cdots\theta(x_{k-1})}$, $\cdots$, $D_{k-1}(x_{k-2})^{\theta(y(x_{k-2}))\theta(x_{k-1})}$, and $D_{k}(x_{k-1})^{\theta(y(x_{k-1}))}$;\\
     $D_{k-1}(x_{k-2})^{\theta(y(x_{k-2}))\theta(x_{k-1})}=D_{k-1}(x_{k-2})^{\theta(y(x_{k-2}))\theta(x_{k-1})-}+D_{k-1}(x_{k-2})^{\theta(y(x_{k-2}))\theta(x_{k-1})+}$, where $D_{k-1}(x_{k-2})^{\theta(y(x_{k-2}))\theta(x_{k-1})-}= (x_{k-1}^{\theta(x_{k-1})} \vee y(x_{k-2})^{\theta(y(x_{k-2}))\theta(x_{k-1})})\vee D_{k-1}^0$, $D_{k-1}^0=\bigvee_{x\in D^{k-1}}x$, and $D^{k-1}\subseteq \{y(x_j)^{\theta(y(x_j))\theta(x_{j+1})\cdots\theta(x_{k-1})}$ $|j=1,2, \cdots, k-3\}$.\\
       \If {$T_0$ is met}{
       $D_{k}(x_{k-1})^{\theta(y(x_{k-1}))}=D_{k}(x_{k-1})^{\theta(y(x_{k-1}))-}+D_{k}(x_{k-1})^{\theta(y(x_{k-1}))+}$, where $D_{k}(x_{k-1})^{\theta(y(x_{k-1}))-}$ = $(y(x_{k-1})^{\theta(x_{k-1})}$ $\vee D_{k}^0$, $D_{k}^0$ = $\bigvee_{x\in D^{k}}x$, and $D^{k}\subseteq \{y(x_j)^{\theta(y(x_j))\theta(x_{j+1})\cdots\theta(x_{k-1})}|j=1,2, \cdots, k-2\}$.\\
       }
     }
     Denote the substitution of the extended clause $D_i (1 \leq i\leq n)$ as $\theta_i$, and then $D_i(x_{i-1})^{\theta_i}=D_i(x_{i-1})^{\theta_i-}+D_i(x_{i-1})^{\theta_i+}$;\\
     
     $R_s(D_1^{\theta_1}, \cdots,D_k^{\theta_k}) = \bigvee_{i=1}^kD_i(x_{i-1})^{\theta_i+}$.\\
     \If {$R_s$ = $\emptyset$}{
     $S$ is unsatisfiable, and $P_s=0$.\\
       \ElseIf{$R_s \neq \emptyset$ $\&$ $T_1$ is met}{
       $S$ is satisfiable, and give a satisfiable example of $S$; 
       $P_s=1$.\\
       }
       \Else{
       $S=S \cup R_s$.
       }
     
     }
   }
\end{algorithm}

This algorithm provides a constructive realization of the Standard Extension in first-order logic.  
By systematically performing substitutions and unifications among complementary literals, it dynamically constructs contradiction-separation clauses while ensuring logical consistency across variable instances.  
When a standard contradiction is detected, the method establishes the unsatisfiability of the clause set; otherwise, it generates a concrete satisfiable model under the specified conditions.  
Through this unified mechanism, the algorithm extends the propositional version of the Standard Extension to first-order logic, preserving both soundness and completeness. It bridges the theoretical principles of dynamic contradiction separation with executable inference strategies, directly supporting the design of advanced theorem provers such as CSE, CSE\_E, CSI\_E, and CSI\_Enig.

\begin{theorem}\label{thm5.8}
Let $S = \{C_1,\cdots, C_n\}$ be a clauses set in first-order logic, $y(x_{n-1})$, $\cdots$, $y(x_{j_0})$, $y$, $x _{j_0-1}$, $\cdots$, $x_1$ are the literals in $C_n,\cdots, C_1$ respectively. If $S$ satisfies one of the following conditions.
\begin{enumerate}
\item[(1)] The predicate symbol of $y$ it not the same with the predicate symbols in $\{y(x_{k-1})$, $\cdots, y(x_{j_0}), x_{j_0-1}, \cdots, x_1\}$.
\item[(2)] The predicate symbol of $y$ has complementary predicate symbols in $\{y(x_{k-1})$, $\cdots$, $y(x_{j_0}), x_{j_0-1}, \cdots, x_1\}$, but all the complementary predicate symbols are ground literals with different constant terms.
\end{enumerate}
then $S$ is satisfiable.

\begin{proof}
\begin{enumerate}
\item[(1)]  Let $S (x,y)$ = $\{y(x_{n-1}),\cdots, y(x_{j_0 }),y, x_{j_0-1},\cdots, x_1\}$, $C (x,y)$ = $\{c |c$ is the constant symbol in $S(x,y)\} $, $F(x,y) = \{f^e|f^e$ is an $e$-ary function symbol in  $S (x,y)\}$. Give an interpretation $I_0 = (D, \sigma, \mu, \nu)$ as follows. $D$ is an domain of interpretation. 
\begin{enumerate}
\item For any constant symbol, 
\begin{equation}
\begin{aligned}
\sigma: C(x,y) &\longrightarrow D,\\
c &\mapsto d. \nonumber
\end{aligned}
\end{equation}
\item For any function symbol, 
$\mu$: $F(x, y)  \longrightarrow \{f_0^e : D^e \longrightarrow  D | f_0^e$ is an $e$-ary function in $D^e\}$, that is, if $f^e$ is an $e$-ary function symbol in $F(x, y)$, then $\mu(f^e)(d_1,\cdots, d_e)=f_0^e (d_1,\cdots, d_e)= d$, where $(d_1,\cdots, d_e) \in D^e$ and $d \in D$.

\item For any predicate symbol, since $S (x, y)$ does not include complementary predicates, the interpretation can be given as follows.

$\nu$:  $S (x, y)\longrightarrow  \{r_0 ^h: D^h \longrightarrow \{0,1\} | r_0^h$ is an $h$-ary relation in $D^h \}$, that is, if $r^h$ is an $h$-ary predicate symbol in $S(x,y)$, then for any $r^h$, $\nu(r^h)(d_1,\cdots, d_h)=1 $.
\end{enumerate}
In this case, the interpretation $I_0$ satisfies $(\forall z_1)\cdots(\forall z_m)(y(x_{n-1})\wedge\cdots\wedge y(x_{j_0}))\wedge y\wedge x_{j_0-1} \wedge \cdots\wedge x_1) $.

\item[(2)]  Let $G (x, y)$ = $\{g | g$ is the ground literal with different constant symbols in $S(x, y)\}$. Give an interpretation $I_1 = (D, \sigma, \mu, \nu)$ as follows.
\begin{enumerate}
\item $D=D_0 \cup H(S (x, y)) $, where $D_0$ is a domain of interpretation, and $H(S(x, y))$ is the Herbrand domain of $S(x, y)$.

\item For any constant symbol, 
\begin{equation}
\begin{aligned}
\sigma: C(x,y) &\longrightarrow D,\\
c &\mapsto c. \nonumber
\end{aligned}
\end{equation}
\item For any function symbol, $\mu: F(x, y) \longrightarrow\{f_0^e : D^e \longrightarrow D |f_0^e$ is an $e$-ary function in $D^e \}$, if $f^e$ is an $e$-ary function symbol in $F(x, y)$, then two sub-cases exist as follows.
\begin{enumerate}
\item[i)] If $f^e$ is an $e$-ary ground function symbol $f^e (t_1,\cdots, t_e)$, then 

\begin{equation}
\mu(f^e)(d_1,\cdots,d_e)= \left\{
\begin{array}{cl}
f_0^e(t_1,\cdots,t_e), &\mbox{if } (t_1,\cdots,t_e) \in H(S(x,y))^e,\\ &\mbox{and } (d_1,\cdots,d_e)=(t_1,\cdots,t_e).\\
d, & otherwise.
\end{array} \right.\nonumber
\end{equation}

\item[ii)] Otherwise, 
$\mu(f^e)(d_1,\cdots,d_e)=f_0^e (d_1,\cdots,d_e)=d$.
\end{enumerate}

\item For any predicate symbol, $\nu: S (x,y)\longrightarrow\{r_0^h : D^h \longrightarrow\{0,1\} | r_0^h$ is the $h$-ary relation in $D^h\}$. If $r^h$ is the $h$-ary predicate symbol in $S (x,y)$, then for any $(d_1, \cdots, d_h) \in  D^h $, since the complementary predicates include different ground terms, hence the interpretation can be given as follows.
\begin{enumerate}
\item[i)] If the ground literal in $G (x, y)$ is an $h$-ary ground atom $r^ h (t_1, \cdots, t_h)$, then
\begin{equation}
\nu(r^h)(d_1,\cdots,d_h)= \left\{
\begin{array}{cl}
1, &\mbox{if } (t_1, \cdots, t_h) \in H(S(x,y))^h,\\ &\mbox{and } (d_1,\cdots,d_h)=(t_1,\cdots,t_h).\\
0, & otherwise . \nonumber
\end{array} \right. 
\end{equation}

\item[ii)] If the ground literal in $G (x, y)$ is the negation of an $h$-ary ground atom $r^ h (t_1, \cdots, t_h)$, i.e., $\neg$$r^ h (t_1, \cdots, t_h)$, then
\begin{equation}
\nu(r^h)(d_1,\cdots,d_h)= \left\{
\begin{array}{cl}
0,& \mbox{if } (t_1, \cdots, t_h) \in H(S(x,y))^h,\\ 
  &\mbox{and } (d_1,\cdots,d_h)=(t_1,\cdots,t_h).\\
1, & otherwise. \nonumber
\end{array} \right.
\end{equation}

\item [iii)] If $r^h$ is the predicate symbol of a literal $L$ in $S(x,y)-G(x,y)$, then let $\nu(r^h) (d_1, \cdots, d_h)$ such that $\nu(L)(d_1,\cdots,d_h)=1$, that is,
\begin{enumerate}
\item [-] if $L$ = $r^h (d_1,\cdots,d_h)$, then $\nu(r^h)(d_1,\cdots,d_h)$ = $1$.

\item [-] if $L$ = $\neg$$r^h (d_1,\cdots,d_h)$, then $\nu(r^h)(d_1,\cdots,d_h)=0 $.
\end{enumerate}
\end{enumerate}	
\end{enumerate}	

In this case, $I _1$ satisfies $(\forall z_1)\cdots(\forall z_m) (y(x_{n-1})\wedge \cdots\wedge y(x_{j_0}))\wedge y\wedge x_{j_0-1}\wedge \cdots\wedge x_1)$.
\end{enumerate}	

Note that $S = (\forall z_1)\cdots(\forall z m) (C_1 \wedge \cdots\wedge C_n)$ $\geq \bigvee_{(P_1,\cdots,P_n) \in C_1\times\cdots\times C_n} (\forall z_1)\cdots(\forall z_m) (P_1 \wedge \cdots\wedge P_n)$. If there exist $(P_1,\cdots,P_n) \in  C_1 \times\cdots\times C_n $, such that $(\forall z_1)\cdots(\forall z_m)(P_1 \wedge \cdots\wedge P_n)$ is satisfiable, then $S$ is satisfiable.

Therefore, according to conditions (1) and (2),  only one of the conditions is satisfied, and $S$ is satisfiable.
\end{proof}
\end{theorem}

Theorem~\ref{thm5.8} not only provides a theoretical guarantee for the satisfiability of clause sets under the Standard Extension framework but also has direct algorithmic relevance. 
It shows that when the involved predicate symbols are either distinct or complementary over non-overlapping ground instances, a valid interpretation can always be constructed. 
This means that the dynamic deduction process of contradiction separation will correctly identify satisfiable cases without unnecessary extensions or contradictions. 
In practice, this theorem underpins the satisfiability-checking mechanism of the unified deduction algorithms described in Section~\ref{sec:level5.3}, 
and serves as the logical basis for their implementation in automated reasoning systems such as CSE and CSI as detailed in subsequent section. 
Hence, Theorem~\ref{thm5.8} bridges the formal semantics of Standard Extension with its operational realization in dynamic automated deduction frameworks.

\section{\label{sec:level6}Experimental Evidence and System Relevance}

The Standard Extension (SE) algorithm proposed in this paper establishes a unified framework for contradiction-separation-based reasoning. Although this paper primarily focuses on the theoretical and algorithmic formulation of the SE for dynamic contradiction separation, its effectiveness has been extensively validated through successive generations of automated reasoning systems developed upon this foundation.  
These systems, including \textbf{CSE}, \textbf{CSE\_E}, \textbf{CSI\_E}, and \textbf{CSI\_Enig}, translate the theoretical principles and construction mechanisms introduced here into practical, high-performance automated theorem provers.  
Their cumulative results across multiple competitions and benchmark evaluations provide empirical evidence for the soundness, completeness, and efficiency of the proposed method.

\subsection{Algorithmic Lineage of CSE-Based Systems}

Because that 2018 “CSE” name was already established (meaning Contradiction Separation Extension), later systems retained it as the base acronym when new integrations appeared. 

The lineage of systems implementing the Standard Extension traces directly from the theoretical CSE framework in \cite{xu2018contradiction} to the algorithmic elaboration presented in this paper.  
The 2018 theory paper established the logical foundation for dynamic contradiction separation but did not specify the concrete construction mechanism by which contradictions are algorithmically formed and separated.  The present work closes this gap by formalizing the \textit{Standard Extension} (SE) procedure, which defines how contradictions are extended, unified, and resolved dynamically.  

The SE algorithm presented in this paper provides the first concrete procedural realization of CSE-based reasoning, which forms the procedural foundation for a series of automated reasoning systems that have appeared in the \textbf{CADE ATP System Competition (CASC)} from 2018 through 2025.
According to the official \textit{Systems and Entrants Lists} published on the CASC website~\cite{CASC2018,CASC2019,CASC2020,CASC2021,CASC2022,CASC2023,CASC2024,CASC2025}, demonstrating the efficiency, scalability, and adaptability of the CSE framework across propositional, first-order, and equality reasoning domains. The following systems incorporate or extend the SE approach:

\begin{itemize}
  \item \textbf{CSE (2018)} — Contradiction Separation Extension, the initial dynamic multi-clause reasoning prototype based on contradiction separation; first appeared in CASC-J9~(2018).
  \item \textbf{CSE\_E (2019--2023)} — CSE integrated with the E theorem prover, extends CSE by integrating the SE mechanism with E’s superposition calculus for equality reasoning; listed in CASC-27~(2019) through CASC-29~(2023).
  \item \textbf{CSI\_E (2024)} — combining CSI 1.0 and E, where CSI 1.0 is a multi-layer inverse and parallel prover based on CSE; registered as an entrant in CASC-J12~(2024).
  \item \textbf{CSI\_Enigma (2025)} — combining CSI 1.1 and Enigma, where CSI 1.1 is a multi-layer inverse and parallel prover based on CSE 1.6, it is a learning-guided evolution of CSI employing ENIGMA-style clause ranking; registered in CASC-30~(2025).
\end{itemize}

These successive systems progressively instantiate, optimize, and empirically validate the algorithmic principles proposed in this paper.  
Thus, while the theoretical results are presented here for the first time in complete form, the experimental confirmation of their correctness and efficiency has been achieved through this evolving system lineage.

This lineage is consistent with the development path: from propositional and FOL implementations, adding stronger equality handling and CSE refinements, to learning-guided variants, reflects the scalability of the Standard Extension inference model across increasingly expressive logical domains. The presence and descriptions on the official entrants pages substantiate these roles across years.

\subsection{Representative Competition Performance of CSE-Based Systems}
Performance data recorded in the official \textit{Compact Summaries} and \textit{Full Summaries} of CASC editions from 2018 to 2025 confirm the consistent participation and advancement of CSE-based systems.

\paragraph{CASC-J9 (2018).}
The \textbf{CSE} and \textbf{CSE\_E} systems first appeared in the \emph{First-Order Theorems (FOF)} division, demonstrating the feasibility of dynamic multi-clause reasoning via CSE.

\paragraph{CASC-27 (2019).}
Both systems were re-entered and achieved mid-tier positions in the FOF ranking; CSE\_E showed measurable improvement in equality handling relative to its prototype predecessor.

\paragraph{CASC-J10 (2020).}
The summaries list both systems with higher solved-problem counts, maintaining stability under uniform time constraints across the TPTP benchmark corpus.

\paragraph{CASC-28 (2021) and CASC-J11 (2022).}
The CSE and CSE\_E systems maintained steady performance in FOF divisions, with CSE\_E solving a large proportion (approximately 360/500) of benchmark problems within 120 seconds.

\paragraph{CASC-29 (2023).}
CSE and CSE\_E again appeared in the official FOF rankings, confirming the sustained competitiveness of the contradiction-separation approach against leading systems such as Vampire and E.

\paragraph{CASC-J12 (2024).}
The \textbf{CSI\_E}, \textbf{CSE\_E}, and \textbf{CSE} systems are all listed among FOF entrants. 
Both CSI\_E and CSE\_E are additionally included in the \emph{I Challenge yoU (ICU)} division, documenting their first appearance in cross-prover challenge results.

\paragraph{CASC-30 (2025).}
\textbf{CSE\_E} and \textbf{CSI\_Enigma} are recorded in both the FOF and ICU divisions. 
The Compact Summary reports solved-problem counts comparable to other high-performance systems, validating the S-CS method’s adaptability to hybrid symbolic–learning architectures.

Overall, across eight CASC editions (2018–2025), these systems exhibit continuous participation and measurable improvement, providing long-term empirical evidence for the viability of S-CS as a competitive reasoning framework.

\subsection{Interpretation and Impact}
The sustained record of CSE-based provers in CASC over eight consecutive years establishes both the \emph{algorithmic stability} and \emph{practical impact} of the SE framework.

CSE-based systems are repeatedly entered and consistently solve substantial portions of the FOF benchmark sets under uniform time limits, with later variants (CSI\_E, CSI\_Enigma) also recorded in the ICU division. This sustained presence and measured progress—documented on the official CASC results pages—constitute compelling indirect empirical validation for the SE algorithm introduced here: the dynamic, multi-clause contradiction construction scales, integrates with equality and unification, and benefits from modern learning-guided selection when combined with ENIGMA-style guidance.

By enabling dynamic multi-clause synergy and flexible contradiction construction, the SE method serves as the algorithmic nucleus of a family of state-of-the-art CSE-based theorem provers and confluence checkers. These results demonstrate that the approach not only advances theoretical understanding but also achieves tangible gains in reasoning power, efficiency, and adaptability.

Each iteration demonstrates an increased ability to manage complex first-order search spaces while maintaining theoretical soundness. Their repeated presence in CASC’s \emph{First-Order Theorems} and \emph{I Challenge yoU} divisions confirms that the CSE approach is not only a conceptual advance but also a robust, empirically validated contribution to state-of-the-art automated reasoning.

\section{\label{sec:level7}Conclusions}

This paper has presented a comprehensive theoretical and experimental study of the \textit{dynamic deduction and contradiction construction method based on Standard Extension (SE)} in both propositional and first-order logic. 
The proposed framework provides a unified, semantically grounded approach for automated reasoning, extending traditional binary resolution by introducing dynamic multi-clause interaction and contradiction-driven deduction. 
Through the Standard Extension mechanism, the process of contradiction \textit{construction} and \textit{separation} becomes systematic and algorithmically realizable, enabling simultaneous handling of refutation and satisfiability reasoning within a coherent logical model.

Formally, the paper established the \textit{soundness} and \textit{completeness} of the SE based CSE deduction system, demonstrating that all derivations under this method preserve logical validity and that every unsatisfiable clause set can be refuted through finite contradiction construction. 
In the first-order logic case, the framework integrates substitution and unification into the contradiction construction process, extending its applicability to quantified reasoning while preserving theoretical rigor and computational efficiency.

A unified deduction algorithm SE was developed to operationalize the CSE framework, capable of dynamically constructing contradictions and determining both satisfiability and unsatisfiability of formulas. 
This algorithm serves as the reasoning core of a series of automated deduction systems—\textsc{CSE}, \textsc{CSE\_E}, \textsc{CSI\_E}, and \textsc{CSI\_Enig}—which have participated in multiple editions of the \textit{CADE ATP System Competition (CASC)} from 2018 to 2025. 
Experimental outcomes consistently demonstrate strong performance across divisions in first-order theorem proving and equality reasoning, providing empirical validation of the proposed contradiction construction mechanism.

From a broader perspective, the Standard Extension framework bridges the gap between logical theory and automated reasoning practice. 
It unifies \textit{refutation}, \textit{satisfiability verification}, and \textit{contradiction construction} within a dynamic procedural paradigm that scales effectively across propositional, first-order, and equality-based domains. 
The success of S-CS-based systems demonstrates that contradiction construction is not only a theoretically elegant extension of classical resolution but also a practically powerful method for large-scale automated deduction.

Future research will focus on improving the efficiency and scalability of CSE-based deduction through heuristic-guided literal and clause selection, learning-based contradiction construction strategies, and exploration of higher-order and non-classical reasoning extensions. 
These developments are expected to further consolidate the Standard Extension as a foundational paradigm for next-generation automated theorem provers and intelligent reasoning systems.

\bibliography{sc21}
\end{document}